\documentclass[envcountsect,envcountsame,final,orivec]{llncs}
\usepackage{times}

\usepackage{tikz}
\usetikzlibrary{arrows,snakes,backgrounds}
\usepackage{stackrel}
\usepackage{hyperref}
\usepackage{ucs}
\usepackage[latin1]{inputenc}
\usepackage{amsmath}
\usepackage{amssymb,amsfonts,amsxtra}
\usepackage{stmaryrd}
\usepackage[mathcal,mathscr]{euscript}
\usepackage[english]{babel}
\usepackage{times}
\usepackage{afterpage}
\usepackage{newvbtm}
\makeatletter
\newif\if@restonecol
\makeatother

\usepackage[boxed,titlenotnumbered,figure]{algorithm2e}

\spnewtheorem{hypothesis}[theorem]{Hypothesis}{\bfseries}{\itshape}
\spnewtheorem{observation}[theorem]{Observation}{\bfseries}{\rmfamily}
\spnewtheorem{example2}[theorem]{Example}{\itshape}{\rmfamily}
\spnewtheorem{definition2}[theorem]{Definition}{\bfseries}{\rmfamily}
\spnewtheorem{remark2}[theorem]{Remark}{\bfseries}{\itshape}
\spnewtheorem{fact}[theorem]{Fact}{\bfseries}{\itshape}

\def\ok#1{\mbox{\raisebox{0ex}[1ex][1ex]{$#1$}}}
\def \tuple#1{\langle #1 \rangle}

\newcommand{\sra}{{\shortrightarrow}}

\newcommand{\ECTL}{\ensuremath{\mathrm{ECTL}}}
\newcommand{\CTL}{\ensuremath{\mathrm{CTL}}}

\newcommand{\CTLS}{\ensuremath{\mathrm{CTL\!}^*}}

\newcommand{\CTLX}{\ensuremath{\mathrm{CTL}\mbox{-}\mathrm{X}}}

\newcommand{\ECTLXG}{\ensuremath{\mathrm{ECTL}\mbox{-}\{\mathrm{X},\mathrm{G}\}}}
\newcommand{\CTLSX}{\ensuremath{\mathrm{CTL\!}^*\mbox{-}\mathrm{X}}}
\newcommand{\id}{\mathrm{id}}
\newcommand{\TS}{\mathrm{TS}}

\newcommand{\pos}{{\mathbf{pos}}}
\newcommand{\ud}{\triangleq}

\newcommand{\ra}{\rightarrow}
\newcommand{\raee}{\sra^{\exists}}
\newcommand{\tle}{\trianglelefteq}

\newcommand{\Lra}{\Leftrightarrow}
\newcommand{\Ra}{\Rightarrow}
\newcommand{\La}{\Leftarrow}

\newcommand{\cK}{{\mathcal{K}}}

\newcommand{\cS}{{\mathcal{S}}}

\newcommand{\cP}{{\mathcal{P}}}

\newcommand{\Pstsim}{\ensuremath{P_{\mathrm{stsim}}}}

\newcommand{\Rstsim}{\ensuremath{R_{\mathrm{stsim}}}}
\newcommand{\SSA}{\ensuremath{\mathit{SSA}}}
\newcommand{\BasicSSA}{\ensuremath{\mathit{BasicSSA}}}

\def\grasse#1{[\![#1]\!]}

\newcommand{\tl}{\vartriangleleft}

\DeclareMathOperator{\StSim}{StSim}
\DeclareMathOperator{\Refiner}{Refiner}
\DeclareMathOperator{\refiner}{Refiner}

\DeclareMathOperator{\bottomBlock}{bottomBlock}
\DeclareMathOperator{\localBottom}{localBottom}
\DeclareMathOperator{\bottom}{Bottom}
\DeclareMathOperator{\Bottom}{Bottom}

\DeclareMathOperator{\Rel}{\mathit{Rel}}

\DeclareMathOperator{\pre}{pre}

\DeclareMathOperator{\Part}{Part}

\DeclareMathOperator{\Split}{\mathit{Split}}

\DeclareMathOperator{\parent}{parent}

\newcommand{\PR}{\KwSty{PR}}
\newcommand{\cbool}{\KwSty{bool}\ }
\newcommand{\cstate}{\KwSty{State}}
\newcommand{\cblock}{\KwSty{Block}}

\newcommand{\cint}{\KwSty{int}\ }
 
\newcommand{\creturn}{\KwSty{return}\ }

\newcommand{\cmaybe}{\KwSty{maybe}}

\newcommand{\ctt}{\KwSty{tt}}
\newcommand{\cff}{\KwSty{ff}}
\newcommand{\cnull}{\KwSty{null}}

\newcommand{\cvoid}{\KwSty{void}\ }

\def\ctemplate#1#2{#1$\langle$#2$\rangle$}
\def\clist#1{\ctemplate{\KwSty{list}}{\KwSty{#1}}}

\def\cmatrix#1{\ctemplate{\KwSty{matrix}}{\KwSty{#1}}}

\begin{document}

\title{Computing Stuttering Simulations}
\author{Francesco Ranzato~~~ Francesco Tapparo}
\institute{Dipartimento di Matematica Pura ed Applicata\\
        Universit\`a di Padova, Italy\\
%\{ranzato,~tapparo\}@math.unipd.it
}

\pagestyle{plain}
\date{}
\maketitle

\begin{abstract}
  Stuttering bisimulation is a well-known behavioral equivalence 
  that preserves $\CTLX$, namely $\CTL$ without the
  next-time operator $\mathrm{X}$. Correspondingly, the stuttering
  simulation preorder induces a coarser behavioral equivalence that
  preserves the existential fragment $\ECTLXG$, namely $\ECTL$ without
  the next-time $\mathrm{X}$ and globally $\mathrm{G}$ operators.
  While stuttering bisimulation equivalence can be computed by the
  well-known Groote and Vaandrager's [1990] algorithm, to the best of
  our knowledge, no algorithm for computing the stuttering simulation
  preorder and equivalence is available.  This paper presents such an
  algorithm for finite state systems.
\end{abstract}

\section{Introduction}
\paragraph{\textbf{The Problem.}} 
Lamport's criticism \cite{lam83} 
of the next-time operator $\mathrm{X}$
in $\CTL$/$\CTLS$ arouse the interest in studying temporal
logics like $\CTLX$/$\CTLSX$, 
obtained  from $\CTL$/$\CTLS$ by removing the next-time operator, and 
related notions of behavioral \emph{stuttering}-based
equivalences~\cite{bcg88,dnv95,gv90}.  
We are interested here in
\emph{divergence blind stuttering} simulation and bisimulation, that
we call, respectively, stuttering simulation and bisimulation for
short.
We focus here on systems specified as Kripke structures (KSs), but analogous 
considerations  hold for labeled transition systems (LTSs).  
Let $\cK = \tuple{\Sigma,\sra,\ell}$ be a KS where
$\tuple{\Sigma,\sra}$ is a transition system and $\ell$ is a
state labeling function. A relation $R\subseteq \Sigma\times \Sigma$
is a stuttering simulation on $\cK$ when for any
$s,t\in \Sigma $ such that $(s,t)\in R$: (1) $s$ and $t$ have the same
labeling by $\ell$ and (2) if  
$\ok{s\sra s'}$ then $t\sra^* t'$ for some $t'$ in such a way that the following
diagram holds:
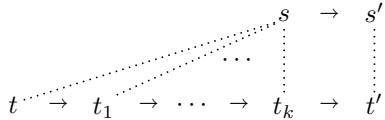
\begin{center}
    \begin{tikzpicture}[scale=0.6]
      \tikzstyle{arrow}=[->,>=latex']
      \draw (6,2) node {$s$};
      \draw (7,2) node {$\sra$};
      \draw (8,2.06) node {$s'$};
      \draw (0,0) node {$t$};
      \draw (5,1) node {$\cdots$};
      \draw (1,0) node {$\sra$};
      \draw (2,-0.06) node {$t_1$};
      \draw (3,0) node {$\sra$};
      \draw (4,0) node {$\cdots$};
      \draw (5,0) node {$\sra$};
      \draw (6,-0.06) node {$t_k$};
      \draw (7,0) node {$\sra$};
      \draw (8,0.03) node {$t'$};
      \path
      (6,1.9) node[name=62]{}
      (8,1.9) node[name=82]{}
      (0,0.05) node[name=00]{}
      (2,0.1) node[name=20]{}
      (6,0) node[name=60]{}
      (8,0) node[name=80]{};

      \draw[dotted, semithick] (62) -- (00);
      \draw[dotted, semithick] (62) -- (20);
      \draw[dotted, semithick] (62) -- (60);
      \draw[dotted, semithick] (82) -- (80);
\end{tikzpicture}
\end{center}
\noindent
where a dotted line between two states means that they are related
by $R$. The intuition is that $t$ is allowed to 
simulate a transition $s\sra s'$ possibly through  
some initial ``stuttering'' transitions ($\tau$-transitions in case of LTSs). 
$R$ is called a stuttering bisimulation when it is symmetric.
It turns out that the largest stuttering simulation $R_{\mathrm{stsim}}$ and 
bisimulation $R_{\mathrm{stbis}}$ relations 
exist: $R_{\mathrm{stsim}}$ is a preorder called the
\emph{stuttering simulation preorder} while $R_{\mathrm{stbis}}$ is an
equivalence relation called the stuttering bisimulation equivalence. Moreover, 
the preorder $R_{\mathrm{stsim}}$ induces by symmetric reduction the
\emph{stuttering simulation equivalence} 
$R_{\mathrm{stsimeq}}=R_{\mathrm{stsim}}\cap
R^{-1}_{\mathrm{stsim}}$. The partition of $\Sigma$
corresponding to the equivalence $R_{\mathrm{stsimeq}}$ is denoted
by $P_{\mathrm{stsim}}$. 

De Nicola and Vaandrager \cite{dnv95} showed that for finite KSs
 and for an interpretation of universal/existential path
quantifiers over all the, possibly finite, prefixes, the stuttering
bisimulation equivalence coincides with the state equivalence induced
by the language $\CTLX$ (this also holds for $\CTLSX$).  This is not
true with the standard interpretation of path quantifiers over
infinite paths, since this requires a divergence sensitive notion of
stuttering (see the details in \cite{dnv95}). Groote and
Vaandrager~\cite{gv90} 
designed a well-known algorithm that computes
the stuttering bisimulation equivalence $R_{\mathrm{stbis}}$ in
$O(|\Sigma| |\sra|)$-time and $O(|\sra|)$-space.

Clearly, stuttering simulation equivalence
is coarser than stuttering bisimulation, i.e.\
$R_{\mathrm{stbis}} \subseteq R_{\mathrm{stsimeq}}$. As far as language
preservation is concerned, it turns out that  stuttering simulation
equivalence coincides with the state equivalence induced
by the language $\ECTLXG$, namely the existiential fragment of $\CTL$
without next-time and globally operators $\mathrm{X}$ and
$\mathrm{G}$. Thus, on the one hand, 
stuttering simulation equivalence still preserves a
significantly expressive 
fragment of $\CTL$ and, on the other hand, 
it may provide a significantly better state space
reduction  than simulation equivalence, and this has been shown to be 
useful in abstract model checking~\cite{man01,ngc05}.

\paragraph{\textbf{State of the Art.}} To the best of our
knowledge, there exists no 
algorithm for computing stuttering simulation
equivalence or, more in general, the stuttering simulation
preorder. There is instead an algorithm by Bulychev et
al.~\cite{bkz07} for \emph{checking} stuttering
simulation, namely, this procedure checks whether a given relation
$R\subseteq \Sigma\times \Sigma$ is a stuttering simulation. This
algorithm formalizes the problem of checking 
stuttering simulation as a two players
game in a straightforward way and then exploits Etessami et
al.'s~\cite{ews01} algorithm for solving such a game. The authors
claim that this provides an
algorithm for checking stuttering simulation on finite KSs
 that runs in
$O(|\sra|^2)$ time and space. 

\paragraph{\textbf{Main Contributions.}}
In this paper we present an algorithm for computing simultaneously 
both the simulation preorder  $R_{\mathrm{stsim}}$ and stuttering
simulation equivalence $R_{\mathrm{stsimeq}}$ for finite KSs. 
This procedure is incrementally designed in two steps. 
We first put forward a basic procedure for computing the stuttering
simulation preorder that relies directly on the notion of stuttering
simulation.  
For any state $x\in \Sigma$, $\StSim(x)\subseteq \Sigma$ represents the
set of states that are candidate to stuttering simulate $x$ so that
a family of sets $\{\StSim(x)\}_{x\in \Sigma}$ is maintained. 
A pair of states $(x,y)\in \Sigma\times \Sigma$ is called a refiner
for $\StSim$ when $x\sra y$ and there exists $z\in \StSim(x)$ that
cannot stuttering simulate $x$ w.r.t.\ $y$, i.e., $z\not\in
\pos(\StSim(x),\StSim(y))$ where $\pos(\StSim(x),\StSim(y))$ is the
set of all the states in $\StSim(x)$ that may reach a state in
$\StSim(y)$ through a path of states in $\StSim(x)$.  Hence, any such 
$z$ can be correctly removed from $\StSim(x)$. 
Actually, it turns out that one such 
refiner $(x,y)$ allows to 
refine $\StSim$ to $\StSim'$ as follows: if
$S=\pos(\StSim(x),\StSim(y))$  then
$$\StSim'(w) := \left\{\begin{array}{ll}
                      \StSim(w)\cap S  & \mbox{~~if $w\in S$}\\
                      \StSim(w) & \mbox{~~if $w\not\in S$}\\
                      \end{array}\right.
$$
Thus, our basic algorithm consists in 
initializing $\{\StSim(x)\}_{x\in \Sigma}$ as $\{ y\in
  \Sigma~|~ \ell(y)=\ell(x)\}_{x\in
  \Sigma}$ and  then iteratively
refining $\StSim$ until a refiner exists.  
This provides an \emph{explicit} stuttering simulation algorithm, meaning
that this procedure requires that for any explicit state $x\in \Sigma$, 
$\StSim(x)$ is explicitly represented as a set of states.

Inspired by techniques used in 
algorithms that compute standard simulation preorders
and equivalences (cf.\ Henzinger et al.~\cite{hhk95} and Ranzato and
Tapparo~\cite{rt07}) and in abstract interpretation-based algorithms
for computing strongly preserving abstract models~\cite{rt07b}, 
our stuttering simulation algorithm
$\SSA$ is obtained by the above basic procedure by exploiting the
following two main 
ideas.

\begin{itemize}
\item[{\rm (1)}] The above explicit algorithm is made ``symbolic'' by
representing the family of sets of states $\{\StSim(x)\}_{x\in
  \Sigma}$ as a family of sets of blocks of a partition $P$ of the
state space $\Sigma$. More precisely, we maintain a partition $P$ of
$\Sigma$ together with a binary relation $\tle \: \subseteq P\times
P$~---~a so-called partition-relation pair~---~so that: (i)~two states
$x$ and $y$ in the same block of $P$ are candidate to be stuttering
simulation equivalent and (ii)~if $B$ and $C$ are two blocks of $P$
and $B\tle C$ then any state in $C$ is candidate to stuttering
simulate each state in $B$. Therefore, here, for any $x\in \Sigma$, if
$B_x\in P$ is the block of $P$ that contains $x$ then
$\StSim(x) =\StSim(B_x)= \cup \{C\in P~|~B_x\tle C\}$. 

\item[{\rm (2)}] In this setting, a refiner of the current partition-relation
$\tuple{P,\tle}$ is a pair of blocks $(B,C)\in P\times P$ 
such that $B\raee C$ and 
$\StSim(B)\not\subseteq \pos(\StSim(B),\StSim(C))$, where 
$\sra^{\exists}$ is the existential
transition relation between blocks of $P$, i.e.,
$B\raee C$ iff  there exist $x\in B$ and $y\in C$ such
that $x\sra y$.
We devise an efficient way for finding a refiner of the current
partition-relation pair that allows us to check whether a given
preorder $R$ is a stuttering simulation in $O(|P||\sra|)$ time and 
$O(|\Sigma||P|\log |\Sigma|)$ space, 
where $P$ is the partition corresponding to the
equivalence $R\cap R^{-1}$. 
Hence, this 
algorithm for checking stuttering
simulation already significantly improves both in time and space 
Bulychev et al.'s~\cite{bkz07} procedure. 
\end{itemize}

Our algorithm $\SSA$ iteratively refines the current partition-relation
pair $\tuple{P,\tle}$ by first splitting the partition $P$ and then by
pruning the relation $\tle$ until a fixpoint is reached. 
Hence, $\SSA$ outputs a  partition-relation
pair \mbox{$\tuple{P,\tle}$} where
$P = P_{\mathrm{stsim}}$ and $y$ stuttering simulates $x$ iff $P(x)
\tle P(y)$, where $P(x)$ and $P(y)$ are the blocks  
of $P$ that contain, respectively, $x$ and $y$. 
As far as complexity is concerned, it turns out that  
$\SSA$ runs in $O(|P_{\mathrm{stsim}}|^2(|\sra| + 
|P_{\mathrm{stsim}}||\sra^{\exists}|))$ time and
$O(|\Sigma||P_{\mathrm{stsim}}| \log |\Sigma|)$ space. 
 It is worth
remarking that stuttering simulation yields a rather coarse equivalence
so that $|\Pstsim|$ should be in general much less than the
size $|\Sigma|$ of the concrete state space. 

%Due to lack of space
%proofs are included in Appendix~\ref{app}. 
%It is known
%\cite{km02} that
%computing standard (non-stuttering) simulations is harder than computing standard
%(non-stuttering) bisumulations.  

\section{Background}\label{background}
 
\paragraph{\textbf{Notation.}}
If $R\subseteq \Sigma\times \Sigma$ is any relation 
%then $R^*\subseteq \Sigma\times \Sigma$ denotes the reflexive
%and transitive closure of $R$ and 
and $x\in \Sigma$ then $R(x)\,\ok{\ud}\, \{x'\in
\Sigma~|~ (x,x')\in R\}$. 
Let us recall that $R$ is called a preorder when it is reflexive and
transitive. 
If $f$ is a function defined on $\wp(\Sigma)$ and $x\in \Sigma$ then
we often write $f(x)$ to mean $f(\{x\})$. 
A partition $P$ of a set $\Sigma$ is a set of nonempty subsets of
$\Sigma$, called blocks, that are pairwise disjoint and whose union
gives $\Sigma$.  $\Part(\Sigma)$ denotes the set of partitions of
$\Sigma$.  If $P\in \Part(\Sigma)$ and $s\in \Sigma$ then $P(s)$
denotes the block of $P$ that contains $s$. $\Part(\Sigma)$ is endowed
with the following standard partial order $\preceq$: $P_1 \preceq
P_2$, i.e.\ $P_2$ is coarser than $P_1$, iff
$\forall B\in P_1. \exists B' \in P_2 . \; B \subseteq B'$.  
%If $P_1,P_2\in \Part(\Sigma)$, $P_1\preceq P_2$ and $B\in P_1$ then
%$\parent_{P_2}(B)$ (when clear from the context the subscript $P_2$
%may be omitted) denotes the unique block in $P_2$ that contains $B$.
For a given nonempty subset $S\subseteq \Sigma$ called splitter, we
denote by $\Split(P,S)$ the partition obtained from $P$ by replacing
each block $B\in P$ with the nonempty sets $B\cap S$ and
$B\smallsetminus S$, where we also allow no splitting, namely
$\Split(P,S)=P$ (this happens exactly when $S$ is a union of some
blocks of $P$). If $B\in P' = \Split(P,S)$ then we denote by
$\parent_P(B)$
(or simply by $\parent(B)$)
the unique block in $P$  that contains $B$ (this 
may possibly be $B$ itself). 
%f $P_1,P_2\in \Part(\Sigma)$, $P_1\preceq P_2$ and $B\in P_1$ then
%$\parent_{P_2}(B)$ (when clear from the context the subscript $P_2$
%may be omitted) denotes the unique block in $P_2$ that contains $B$.
%
\\
\indent
A transition system $(\Sigma ,\sra)$ consists of a 
set $\Sigma$ of states and a transition relation $\sra 
\subseteq \Sigma \times
\Sigma$. 
%The relation $\sra$ is total when
%for any $s\in \Sigma$ there exists some $t\in \Sigma $ such
%that $s\sra t$. 
The predecessor
transformer $\pre:\wp(\Sigma)\ra \wp(\Sigma)$  is
defined as usual:
$\pre (Y) \,\ok{\ud}\, \{ s\in \Sigma ~|~\exists t\in Y.\; s \sra
t\}$.
If $S_1,S_2\subseteq \Sigma$ then $S_1
\raee S_2$ iff there exist $s_1\in S_1$ and $s_2\in
S_2$ such that $s_1 \sra s_2$.  
Given a set $\mathit{AP}$ of atomic propositions
(of some specification language), a Kripke structure (KS) 
$\cK= (\Sigma ,\sra,\ell)$ over
$\mathit{AP}$ consists of a transition system $(\Sigma ,\sra)$
together with a state labeling function $\ell:\Sigma \ra
\wp(\mathit{AP})$. $P_\ell\in \Part(\Sigma)$ denotes 
the state partition induced by $\ell$, namely,
$P_\ell \,\ok{\ud}\, \{\{s'\in
\Sigma~|~ \ell(s)=\ell(s')\}\}_{s\in \Sigma}$.

\paragraph{\textbf{Stuttering Simulation.}}
Let $\cK = (\Sigma ,\sra ,\ell)$ be a KS. 
A
relation $R \subseteq \Sigma \times \Sigma $ is a \emph{divergence blind 
stuttering simulation}
on $\cK$ if for any
$s,t\in \Sigma $ such that $(s,t)\in R$: 
\begin{itemize}
\item[{\rm (1)}] $\ell(s) =\ell (t)$;
\item[{\rm (2)}] If
$\ok{s\sra s'}$ 
then there exist $t_0,...,t_k\in \Sigma$, with $k\geq 0$, such that: (i)
$t_0=t$; (ii)  for all $i\in [0,k)$, $\ok{t_i
\sra  t_{i+1}}$ and $(s,t_i)\in R$; 
(iii) $(s',t_k)\in R$.
\end{itemize}
Observe that condition~(2) allows the case $k=0$
and this boils down to requiring that $(s',t)\in R$.
With a slight abuse of terminology, 
$R$ is called simply a \emph{stuttering simulation}.    
If $(s,t)\in R$ then we say that $t$ stuttering simulates $s$ and we
denote this by $s \leq t$.  If $R$ is a symmetric relation then it is
called a stuttering bisimulation.  
The empty
relation is a stuttering simulation and stuttering 
simulations are closed under
union so that the 
largest stuttering simulation relation exists. It turns out that the largest
simulation is a preorder relation called \emph{stuttering simulation preorder} (on
$\cK$) and
denoted by $R_{\mathrm{stsim}}$. Thus, for any $s,t\in \Sigma$, $s\leq
t$ iff $(s,t)\in  \ok{R_{\mathrm{stsim}}}$.
\emph{Stuttering simulation equivalence} 
$R_{\mathrm{stsimeq}}$ is the symmetric reduction of 
$R_{\mathrm{stsim}}$, namely $R_{\mathrm{stsimeq}} \,\ok{\ud}\, R_{\mathrm{stsim}}
\cap \ok{R_{\mathrm{stsim}}^{-1}}$, so that $(s,t)\in R_{\mathrm{stsimeq}}$
iff $s\leq t$ and $t\leq s$. 
$P_{\mathrm{stsim}}\in \Part(\Sigma)$ denotes
the partition corresponding to the equivalence 
$R_{\mathrm{stsimeq}}$ and is called
stuttering simulation partition. 
\\
\noindent
Following Groote and Vaandrager~\cite{gv90}, 
$\pos:\wp(\Sigma)\times \wp(\Sigma) \sra \wp(\Sigma)$ is defined as:
\vspace*{-10pt}
\begin{multline*}
\pos(S,T) \triangleq\\
\shoveright{ \{s\in S ~|~\exists k\geq 0. \exists
s_0,...,s_k.\; s_0=s\;\&\; \forall i\in [0,k).\, s_i\in S,\: s_i  \sra
s_{i+1}\; \&\; s_k\in T\} }\\[-20pt]
\end{multline*}
so that  a relation $R \subseteq\Sigma\times\Sigma$ 
  is a stuttering simulation iff for any $x,y\in \Sigma$, 
  $R(x)\subseteq
  P_\ell(x)$ and  if $x\sra y$ then
  $R(x)\subseteq \pos(R(x),R(y))$.
\\
\indent
It turns out \cite{dnv95} that $P_{\mathrm{stsim}}$ is the coarsest
partition preserved by the temporal language $\ECTLXG$. More precisely, 
$\ECTLXG$ is inductively defined as follows:
$$\phi ::= 
~p ~ |~ \neg p ~|~ \phi_1 \wedge \phi_2 ~|~ \phi_1 \vee \phi_2 ~|~ 
\mathrm{EU}(\phi_1,\phi_2)$$
and its semantics is standard: 
$\grasse{p} \,\ok{\ud}\, \{s\in \Sigma~|~p\in \ell(s)\}$ and
$\grasse{\mathrm{EU}(\varphi_1,\varphi_2)} \,\ok{\ud}\,
\grasse{\varphi_2} \cup \pos(\grasse{\varphi_1},
\grasse{\varphi_2})$. 
The coarsest partition preserved by $\ECTLXG$ is the state partition
corresponding to the following equivalence $\sim$ between states: for
any $s,t\in \Sigma$,
$$s\sim t ~~\text{iff}~~ \forall \phi\in \ECTLXG.\; s\in \grasse{\phi}
\Lra t\in \grasse{\phi}.$$

%It is simple to show that $\EU$ is monotone and idempotent on the
%second argument, i.e., for any $S,T\subseteq\Sigma$, 
%$\EU(S,T)=\EU(S,\EU(S,T))$.

\section{Basic Algorithm}

\linesnotnumbered
\begin{algorithm}[Ht]
\small
\printsemicolon
\SetVline
\SetAlTitleFnt{textsc}
\Indm

\FuncSty{{\rm $\BasicSSA (\text{Partition}~ P_\ell)\;\{$}}

\Indp
\lForAll{$x\in \Sigma$}{$\StSim(x) := P_\ell(x)$\;}

\While{$(\exists x,y\in \Sigma ~\KwSty{such that}$ 
$x\sra y \;\&\; \StSim(x) \not\subseteq \pos(\StSim(x),\StSim(y)))$}
{
  $S := \pos(\StSim(x),\StSim(y))$\;
  \lForAll{$w\in S$} {$\StSim(w) := \StSim(w) \cap S$\;} 
}

\Indm
\FuncSty{\}}

\caption{Basic Stuttering Simulation Algorithm $\BasicSSA$.}\label{algouno}
\end{algorithm}

For each state $x\in \Sigma$, the algorithm $\BasicSSA$ 
in Figure~\ref{algouno} 
computes the stuttering simulator set
$\StSim(x)\subseteq \Sigma$, i.e., the set of 
states that stuttering simulate $x$. The basic idea is that
$\StSim(x)$ contains states that are candidate for stuttering
simulating $x$. Thus, the input partition of $\BasicSSA$ is taken as 
the partition $P_\ell$ determined by the labeling $\ell$ so that 
$\StSim(x)$ is initialized with $P_\ell(x)$, i.e.,
with all the states that
have the same labeling of $x$. Following the definition of stuttering
simulation, 
a refiner is a
pair of states $(x,y)$ such that $x \sra y$ and 
$\StSim(x)\not\subseteq  \pos(\StSim(x),\StSim(y))$. In fact, if  
$z\in \StSim(x) \smallsetminus \pos(\StSim(x),\StSim(y))$ 
then $z$ cannot
stuttering simulate $x$ and therefore can be correctly removed from
$\StSim(x)$. 
Conversely, if no such refiner exists then for any
$x,y\in \Sigma$ such that $x\sra y$ we have that 
$\StSim(x) \subseteq  \pos(\StSim(x),\StSim(y))$ so that
any $z\in \StSim(x)$ actually stuttering simulates $x$. 
Hence, $\BasicSSA$ consists in iteratively refining
$\{\StSim(x)\}_{x\in \Sigma}$ as long as a refiner exists, where, 
given a refiner $(x,y)$,  
the refinement of $\StSim$ by means of $S=\pos(\StSim(x),\StSim(y))$ 
is as follows:
$$\StSim(w) := \left\{\begin{array}{ll}
                      \StSim(w)\cap S  & \mbox{~~if $w\in S$}\\
                      \StSim(w) & \mbox{~~if $w\not\in S$}\\
                      \end{array}\right.
$$
It turns
out that this procedure correctly computes the stuttering simulation preorder. 

\begin{theorem}\label{basicssa}
$\BasicSSA$ is correct, i.e., if $\StSim$ is the output 
of $\BasicSSA$ on input $P_\ell$ then for any $x,y\in \Sigma$,
$y\in \StSim(x) \:\Lra\: x\leq y$.  
\end{theorem}

\section{Partition-Relation Pairs}
A \emph{partition-relation pair} $\tuple{P,\tle}$,  PR for short, is given by a 
partition $P\in \Part(\Sigma)$ together with
a binary relation $\tle\; \subseteq P\times P$
between blocks of $P$. 
We write $B\tl C$ when $B\tle C$ and $B\neq C$ and $(B',C')\tle (B,C)$ 
when $B'\tle B$ and $C'\tle C$. 
%If $B\tle C$ then sometimes we say that $B$ is below $C$. 
Our stuttering simulation algorithm relies on the idea of symbolizing
the $\BasicSSA$ procedure in order to maintain a PR $\tuple{P,\tle}$
in place of the family of explicit sets of states $\{\StSim(s)\}_{s\in
\Sigma}$. As a first step,  $\cS= \{\StSim(s)\}_{s\in
\Sigma}$ induces a partition $P$ that corresponds to 
the following equivalence $\sim_\cS$: 
$$s_1
\sim_\cS s_2 \text{~~iff~~} \forall s\in \Sigma. \; s_1 \in \StSim(s) \Lra s_2
\in \StSim(s).$$ 
Hence, the intuition is that 
if $P(s_1)=P(s_2)$ then $s_1$ and $s_2$ are
``currently'' candidates to be stuttering simulation equivalent. 
Accordingly, 
a relation $\tle$ on $P$ encodes stuttering simulation as follows:
if $s\in \Sigma$ then
$\StSim(s) = \{ t\in \Sigma~|~ P(s) \tle P(t)\}$.  Here, the intuition
is that if $B\tle C$ then any state $t\in C$ is ``currently''
candidate to stuttering simulate any state $s\in B$. Equivalently, 
the following invariant property is maintained: if
$s \leq t$ then $P(s) \tle P(t)$.   
Thus, a PR
$\tuple{P,\tle}$ will represent the current approximation of the
stuttering simulation preorder and in particular $P$ will represent the
current approximation of stuttering simulation equivalence.

More precisely, a PR $\cP= \tuple{P,\tle}$ induces the following map
$\mu_\cP : \wp(\Sigma) \ra \wp(\Sigma)$: for any $X\in
\wp(\Sigma)$,
$$\mu_\cP (X) \ud
\cup\{ C\in P ~|~ \exists B\in P.\, B\cap X \neq
\varnothing,\: B\tle C \}.$$ 
Note that, for any $s\in \Sigma$, 
$\mu_\cP (s) =
\mu_\cP (P(s)) = \{ t\in \Sigma~|~ P(s) \tle P(t)\}$, that is, 
$\mu_\cP (s)$ represents the set of states that are currently
candidates to stuttering simulate $s$. 
A  PR  $\cP=\tuple{P,\tle}$ is therefore defined to be a stuttering
simulation for a KS $\cK$ when the relation 
$\{(s,t)\in \Sigma\times \Sigma~|~ s\in \Sigma,\: t\in \mu_\cP(s)\}$
is a stuttering simulation on $\cK$.

%(1)~$P\preceq P_\ell$; 
%(2)~if $t\in \mu_\cP (s)$ then $s\leq t$.
%It turns out that a stuttering simulation $\cP$ 
%can be characterized as follows. 

%It is easy to show that $\mu_\cP$ is indeed
%an additive closure operator on the complete lattice
%$\tuple{\wp(\Sigma),\subseteq}$.

%\begin{lemma}\label{st-char}
%$\cP= \tuple{P,\tle}$ is a stuttering simulation iff $P\preceq
%P_\ell$ and for all $B,C\in P$, 
%$\pos(\mu_\cP(B),\mu_\cP(C)) =
% \mu_\cP(\pos(\mu_\cP(B),\mu_\cP(C)))$. 
%\end{lemma}

Recall that in $\BasicSSA$ 
a pair of states $(s,t)\in \Sigma\times \Sigma$ is a refiner for
$\StSim$ when
$s\sra t$ and
$\StSim(s) \not\subseteq \pos(\StSim(s),\StSim(t))$. 
Accordingly, a
pair of blocks $(B,C)\in P\times P$ is called a refiner 
for  $\cP$ when $B\raee C$ and $\mu_\cP (B) \not\subseteq 
\pos(\mu_\cP (B),\mu_\cP (C))$. Thus, by defining 
$$\Refiner(\cP) \triangleq \{(B,C)\in P^2~|~B\raee C,\; 
\mu_\cP (B) \not\subseteq 
\pos(\mu_\cP (B),\mu_\cP (C))\}$$
the following characterization holds:

\begin{theorem}\label{st-char2}
$\cP=(P,\trianglelefteq)$ is a stuttering simulation 
iff $\refiner(\cP)=\varnothing$ and for any $s\in \Sigma$,
$\mu_\cP(s)\subseteq P_\ell(s)$.  
\end{theorem}

\subsection{A Symbolic Algorithm}\label{asa} 
The algorithm $\BasicSSA$ is therefore made symbolic as follows:

\begin{itemize}
\item[{\rm (1)}] $\tuple{P_\ell,\id}$ is the input PR, where $(B,C)\in
  \id \Lra B=C$; 

\medskip
\item[{\rm (2)}] Find $(B,C)\in \refiner (\cP)$; if
  $\refiner(\cP)=\varnothing$ exit;

\medskip
\item[{\rm (3)}] Compute $S=\pos(\mu_\cP(B),\mu_\cP(C))$;

\medskip
\item[{\rm (4)}] $\cP' := \tuple{P',\tle'}$, where $P' = \Split(P,S)$
  and $\tle'$ is modified in such a way
  that for any $s\in \Sigma$, $\mu_{\cP'}(P'(s))=\mu_\cP(P(s))$;
\medskip
\item[{\rm (5)}] $\cP'' := \tuple{P',\tle''}$, where 
$\tle'$ is modified to $\tle''$ in such a way that for any $B\in P'$:
$$\mu_{\cP''}(B) = \left\{\begin{array}{ll}
                      \mu_{\cP'}(B)\cap S    & \mbox{~~if $B\subseteq S$}\\
                      \mu_{\cP'}(B) & \mbox{~~if $B\cap S=\varnothing$}
                      \end{array}\right.
$$

\item[{\rm (6)}] $\cP := \cP''$ and go to (2).
\end{itemize}

\incmargin{0.5em}
\linesnumbered
\begin{algorithm}[Ht]
\small
\printsemicolon
\SetVline
\Setnlskip{1.5em}
\SetAlTitleFnt{textsc}
\Indm

\FuncSty{$\SSA (\PR\;\tuple{P,\Rel})\;\{$}

\Indp
$\mathit{Initialize()}$\;
\While{ $((B,C) := \mathit{FindRefiner}())$ $\ne$ $(\cnull,\cnull)$}{
  \clist{\cstate} $X := \mathit{Image}(\tuple{P,\Rel},B)$,~~ $Y := \mathit{Image}(\tuple{P,\Rel},C)$\; 
  \clist{\cstate} $S := \pos(X,Y)$\;
  $\mathit{SplittingProcedure}(\tuple{P,\Rel},S)$\;
  $\mathit{Refine}(\tuple{P,\Rel},S)$\;
}

\Indm
\FuncSty{\}}

\caption{Stuttering Simulation Algorithm $\SSA$.}\label{fig-ssa}
\end{algorithm}

%\noindent
This leads to the symbolic algorithm $\SSA$ described in
Figure~\ref{fig-ssa}, where: the input PR $\tuple{P,\Rel}$ at line~1
is $\tuple{P_\ell,\id}$ of point~(1); point~(2) corresponds to the
call $\mathit{FindRefiner}()$ at line~3; point~(3) corresponds to
lines~4-5; point~(4) corresponds to the call
$\mathit{SplittingProcedure}(\tuple{P,\Rel},S)$ at line~6; point~(5)
corresponds to the call
$\mathit{Refine}(\tuple{P,\Rel},S)$ at line~7.
The following graphical example shows how points~(4) and~(5) refine a
PR $\tuple{\{[0,1],[2,3],[4,5],[6,7],[8,9]\},\tle}$ 
w.r.t.\ the set 
$S=\{3,4,5,8\}$, where if $B \tl C$ then $B$ is drawed below $C$
while if $B\tl C$ and $C\tl B$ then $B$ and $C$ are at same height and connected by a double line.
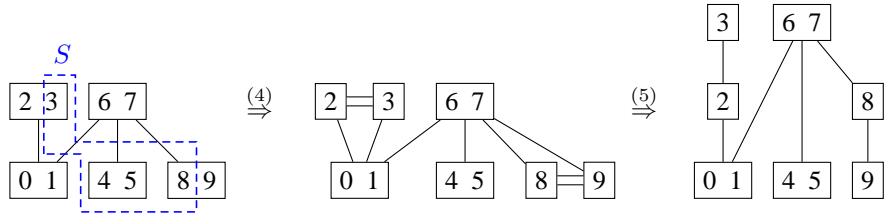
\begin{center}
    \begin{tikzpicture}[scale=0.7]
      \tikzstyle{arrow}=[->,>=latex']
      \path 
      (0,0) node[shape=rectangle,draw,name=01]{0~~1}
      (0,1.5) node[shape=rectangle,draw,name=23]{2~~3}
      (1.5,0) node[shape=rectangle,draw,name=45]{4~~5}
      (1.5,1.5) node[shape=rectangle,draw,name=67]{6~~7}
      (3,0) node[shape=rectangle,draw,name=89]{8~~9}
      (0.45,2.4) node[]{\color{blue}$S$};

      \draw[solid] (01) -- (23);
      \draw[solid] (01) -- (67);
      \draw[solid] (45) -- (67);
      \draw[solid] (89) -- (67);

\path
      (4.2,1.5) node[]{$\stackrel{(4)}{\mathbf{\Rightarrow}}$};

      \path 
      (6.1,0) node[shape=rectangle,draw,name=01m]{0~~1}
      (5.55,1.5) node[shape=rectangle,draw,name=2m]{2}
      (6.65,1.5) node[shape=rectangle,draw,name=3m]{3}
      (10.65,0) node[shape=rectangle,draw,name=9m]{9}
      (9.55,0) node[shape=rectangle,draw,name=8m]{8}
      (8.1,1.5) node[shape=rectangle,draw,name=67m]{6~~7}
      (8.1,0) node[shape=rectangle,draw,name=45m]{4~~5};

	   \draw[solid] (01m) -- (2m);
	   \draw[solid] (01m) -- (3m);
	   \draw[double distance=3pt,solid] (2m) to (3m);
      \draw[solid] (01m) -- (67m);
      \draw[solid] (45m) -- (67m);
      \draw[solid] (8m) -- (67m);
      \draw[solid] (9m) -- (67m);
       \draw[double distance=3pt,solid] (8m) -- (9m);

      \path
      (11.5,1.5) node[]{$\stackrel{(5)}{\mathbf{\Rightarrow}}$};

      \path 
      (13,0) node[shape=rectangle,draw,name=01n]{0~~1}
      (13,1.5) node[shape=rectangle,draw,name=2n]{2}
      (13,3) node[shape=rectangle,draw,name=3n]{3}
      (15.75,0) node[shape=rectangle,draw,name=9n]{9}
      (15.75,1.5) node[shape=rectangle,draw,name=8n]{8}
      (14.5,3) node[shape=rectangle,draw,name=67n]{6~~7}
      (14.5,0) node[shape=rectangle,draw,name=45n]{4~~5};

      \draw[solid] (01n) -- (2n);
      \draw[solid] (01n) -- (67n);
      \draw[solid] (2n) -- (3n);
      \draw[solid] (9n) -- (8n);
      \draw[solid] (8n) -- (67n);
      \draw[solid] (45n) -- (67n);

\path[draw=blue,semithick,densely dashed] (0.1,0.5) -- (0.1,2) -- (0.7,2) -- (0.7,0.73) -- (3,0.73) --
(3,-0.6) -- (0.8,-0.6) -- (0.8,0.5) -- (0.1,0.5);

  \end{tikzpicture}    
\end{center}
The correctness of this symbolic algorithm goes as follows.  

\begin{theorem}[\textbf{Correctness}]\label{correctness}
$\SSA$ is a correct implementation of $\BasicSSA$, i.e., if 
$\StSim$ is the output
function of $\BasicSSA$ on input $P_\ell$ and
$\cP = \tuple{P,\Rel}$ is the output PR of $\SSA$ on input
$\tuple{P_\ell,\id}$ then for any $x\in \Sigma$, $\StSim(x)=
\mu_\cP(x)$. 
\end{theorem}

The next step consists in devising an efficient implementation of $\SSA$.

\section{Bottom States}
\label{sec:bs}

While it is not too hard to devise an efficient implementation of
lines~2 and 4-7 of the $\SSA$ algorithm, 
it is instead not straightforward to find
a refiner in an efficient way. In Groote and Vaandrager's~\cite{gv90}
algorithm for computing stuttering bisimulations the key point for
efficiently finding a refiner in their setting is the notion of \emph{bottom state}.
Given a set of states $S\subseteq \Sigma$,
a bottom state of $S$ is a state
$s\in S$ that cannot go inside $S$, i.e., $s$ can only go outside $S$
(note that $s$ may also have no outgoing transition).  
For any 
$S\subseteq\Sigma$, we therefore define: $$\bottom(S)\triangleq
S\smallsetminus\pre(S).$$

Bottom states allow to efficiently find refiners in KSs
that do not contain cycles of states all having the same labeling.
Following Groote and Vaandrager~\cite{gv90}, a transition $s\sra t$ is
called \emph{inert} for a partition $P\in \Part(\Sigma)$ when
$P(s)=P(t)$. Clearly, if a set of states $S$ in a KS
$\cK$ is strongly
connected via inert transitions for the labeling partition $P_\ell$
then all the states in $S$ are stuttering simulation equivalent, i.e.,
if $s,s'\in S$ then $\Pstsim(s)=\Pstsim(s')$. Thus, each strongly
connected component (s.c.c.) $S$ with respect to inert transitions for
$P_\ell$, called \emph{inert s.c.c.}, can be collapsed to one
single ``symbolic state''. In particular, if $\{s\}$ is one such
inert s.c.c., i.e.\ $s\sra s$, then this collapse simply removes the
transition $s\sra s$. It is important to remark that 
a standard depth-first search algorithm by Tarjan \cite{cormen}, 
running in $O(|\Sigma| + |\sra|)$ time,
allows us to find and then collapse all the inert s.c.c.'s in the
input KS. We can thus assume w.l.o.g.\ that the KS
$\cK$ does not contain inert s.c.c.'s.  The following
characterization of refiners therefore holds.

\begin{lemma} \label{l2}
Assume that $\cK$ does not contain inert s.c.c.'s.
Let $\cP= \tuple{P,\tle}$ be a PR such that for any $B\in P$,
$\mu_\cP(B) \subseteq P_\ell(B)$.  Consider 
$(B,C)\in P\times P$ such that $B\raee C$.
Then, $(B,C)\in \refiner(\cP)$ 
iff $\bottom(\mu_\cP (B))\not\subseteq \mu_\cP(C) \cup \pre(\mu_\cP (C))$.
\end{lemma}

If $B\in P$ is any block then we define as 
\emph{local bottom states} of $B$ 
all the bottom states of $\mu_\cP(B)$ that belong to $B$, namely
$$\localBottom(B) \triangleq \bottom(\mu_\cP(B))\cap B.$$
Also, we define $C\in P$ as a \emph{bottom
  block} for $B$ when $C$ contains at least a bottom state of $\mu_\cP(B)$ and
$B \tl  C$, that is:
$$
\bottomBlock(B) \triangleq
\{C \in P~|~B\tl C,\: C\cap \bottom(\mu_\cP (B)) \neq
  \varnothing\}.
$$
Local bottoms and  
bottom blocks characterize refiners for stuttering simulation as follows:

\begin{theorem}\label{th-refiner}
Assume that $\cK$ does not contain inert s.c.c.'s. 
Let $\cP = \tuple{P,\tle}$ be  a PR such that $\tle$ is a preorder
and for any $B\in P$,
$\mu_\cP(B) \subseteq P_\ell(B)$.
Consider $(B,C)\in P\times P$ such that $B\raee C$ and
for any $(D,E)$ such that $D\raee E$ and 
$(B,C) \tl (D,E)$, $(D,E) \not\in \refiner(\cP)$. 
Then, $(B,C)\in \refiner(\cP)$ iff 
at least one of the following two conditions holds:
  \begin{enumerate}
  \item[{\rm (i)}] $C\not \trianglelefteq B$ and  
$\localBottom(B)\not \subseteq\pre(\mu_\cP (C))$;
  \item[{\rm (ii)}] There exists $D\in \bottomBlock(B)$ such that $C \not 
    \trianglelefteq D$ and  
   $D \not\!\raee \mu_\cP(C)$. 
  \end{enumerate}
\end{theorem}

We will show that 
this characterization provides the basis for 
an algorithm that efficiently finds refiners. Hence, this procedure
also checks whether a given
preorder $R$ is a stuttering simulation. This can be done 
in $O(|P||\sra|)$ time and
$O(|\Sigma||P|\log |\Sigma |)$ space, 
where $P$ is the partition corresponding to the
equivalence $R\cap R^{-1}$. Thus, this 
algorithm for checking stuttering
simulation already significantly improves 
Bulychev et al.'s~\cite{bkz07} procedure that runs in $O(|\sra|^2)$
time and space.

\section{Implementation}\label{sec-impl}

\subsection{Data Structures}

$\SSA$ is implemented by exploiting the following 
data structures. 

\begin{itemize}
\item[{\rm (i)}] A state $s$ is represented by a record that contains 
the list $s$.pre of its
  predecessors $\pre(s)$ and a pointer $s$.block to the block $P(s)$
  that contains $s$. 
  %$\{s.\text{pre}\}$ therefore represents the transition system.  
  The state space $\Sigma$
  is represented as a doubly linked list of states. 

\item[{\rm (ii)}] The states of any block $B$ of the current partition
  $P$ are consecutive in the list $\Sigma$, so that $B$ is represented
  by two pointers begin and end: the first state of $B$ in $\Sigma$ and the
  successor of the last state of $B$ in $\Sigma$, i.e.,
  $B=[B.\text{begin},B.\text{end}[$.  Moreover, $B$ contains a pointer
  $B$.intersection to a block whose meaning is as follows: after a
  call to $\mathit{Split(P,S)}$ for splitting $P$ w.r.t.\ a set of
  states $S$, if $\varnothing \neq B\cap S \subsetneq B$ then
  $B$.intersection points to a block that represents 
$B\cap S$, otherwise $B$.intersection $=
  \textbf{null}$. Finally, the fields localBottoms and
  bottomBlocks for a block $B$ represent, resp., the local bottom
  states of $B$ and the
  bottom blocks of $B$.
The current partition $P$ is stored as a doubly linked list of blocks.  

\item[{\rm (iii)}] The current relation $\tle$ on $P$ is stored as a
  resizable $|P|\times|P|$ boolean matrix $\Rel$: $\Rel(B,C) = \ctt$
  iff $B\tle C$.  Recall \cite[Section~17.4]{cormen} that insert
  operations in a resizable array (whose capacity is doubled as
  needed) take amortized constant time, and a resizable matrix (or
  table) can be implemented as a resizable array of resizable arrays.  
  The boolean matrix $\Rel$ is resized
  by adding a new entry to $\Rel$, namely a new row and a new column,
  for any block $B$ that is split into two new blocks $B\smallsetminus
  S$ and $B\cap S$.

\item[{\rm (iv)}] $\SSA$ additionally stores and maintains 
a resizable integer table Count and a resizable integer matrix BCount.  
Count is indexed over $\Sigma$ and $P$ and has
the following meaning: 
$\text{Count}(s,C)\ud |\{(s,t)~|~D \tle C, t\in D, s\sra t\}|$
while BCount is indexed over $P\times P$ and has the following
meaning:
$\text{BCount}(B,C) \ud
      \textstyle{\sum}_{s\in B} \text{Count}(s,C)$.
The table Count allows to implement the test 
$s\not \in \pre(\mu_\cP (C))$ in constant time  as $\text{Count}(s,C)=0$, 
while BCount allows to implement in constant time the test
$B\not\!\raee \mu_\cP (C)$ as $\text{BCount}(B,C)=0$. 

\end{itemize}

\incmargin{0.2em}
\linesnumbered
\begin{algorithm}[Ht]
\small
\printsemicolon
\SetVline
\Setnlskip{1.5em}
\SetAlTitleFnt{textsc}
\Indm
\KwSty{Precondition: } The list $P$ is stored in reverse topological ordering wrt
$\Rel$

\BlankLine
$\tuple{\KwSty{Block},\KwSty{Block}}$ \FuncSty{$\mathit{FindRefiner}()\; \{$}

\Indp

\cmatrix{\cbool} Refiner\;
\lForAll{$B\in P$}{
  \lForAll{$C\in P$}{Refiner($B$,$C$) $:=$ \cmaybe\;}
}
\ForAll{$C \in P$}{
  \ForAll{$B \in P$ \KwSty{such that} $B\raee C$ }{
    \If{{\rm (Refiner($B$,$C$) $=$ \cmaybe)}}{
      \If{{\rm ($\Rel(C,B) = \cff$)}}{
        \ForAll{{\rm $s\in B$.localBottoms}} {
          \lIf{{\rm (Count$(s,C)=0$)}}{
            \creturn $(B,C)$\;
          }
        }
      }
      \ForAll{{\rm $D \in B$.bottomBlocks}}{
        \lIf{{\rm ($\Rel(C,D) = \cff \;\&\; \text{BCount}(D,C)=0$)}}{\creturn $(B,C)$\;}
      }
    \ForAll{$E\in P$}{
      \lIf{$(\Rel(E,C)=\ctt)$}{Refiner($B$,$E$) $:=$ \cff\;}
    }
  }
}
}
\creturn (\cnull,\cnull)\;

\Indm
\FuncSty{\}}

\caption{$\mathit{FindRefiner}()$ algorithm.}\label{fig-refiner}
\end{algorithm}

\subsection{$\mathit{FindRefiner}$ Algorithm}
The algorithm $\mathit{FindRefiner}()$ in Figure~\ref{fig-refiner} is
an implementation of the characterization of refiners provided by 
Theorem~\ref{th-refiner}. In particular, lines~8-10 implement
condition~(i) of Theorem~\ref{th-refiner} and  lines~11-12 implement 
condition~(ii). The correctness of this implementation depends on the
following key point. Given a pair of blocks $(B,C)\in P\times P$ such
that $B\raee C$, in order to ensure the equivalence: 
$(B,C)\in \Refiner(\cP)$ iff $\text{(i)}\vee \text{(ii)}$,  
Theorem~\ref{th-refiner} requires as hypothesis the following condition: 
$$
\forall (D,E)\in P\times P.\; D\raee E \; \&\; (B,C) \tl (D,E) \; \Ra 
(D,E)\not\in \Refiner(\cP)~~~(*)
$$
In order to ensure this condition $(*)$, we guarantee throughout the
execution of $\SSA$ that the list
$P$ of blocks is stored in reverse topological ordering 
w.r.t.\ $\tle$, so
that if $B \tl B'$ then $B'$ precedes $B$ in the list $P$.
The reverse
topological ordering of $P$ initially holds because the input
PR is  the DAG
$\tuple{P_\ell,\id}$ which is trivially topologically ordered
(whatever the ordering of $P_\ell$ is). 
More in general, for a generic input PR $\tuple{P,\Rel}$ to $\SSA$ 
the function
$\mathit{Initialize}()$ in Figure~\ref{fig-2} in Appendix~\ref{app} achieves 
this reverse topological ordering by a standard algorithm
\cite[Section~22.4]{cormen} that runs in $O(|P|^2)$ time (cf.\
the call $\mathit{TopologicalSort}(P,\Rel)$ in the
$\mathit{Initialize}()$ function).  Then, the reverse topological
ordering of $P$ is always maintained throughout the
execution of $\SSA$. In fact, if the partition $P$ is split w.r.t.\ a
set $S$ and a block $B$ generates two new descendant blocks $B\cap S$
and $B\smallsetminus S$ then our $\mathit{SplittingProcedure}$ in
Figure~\ref{fig-sp} modifies the ordering of the 
list $P$ as follows: $B$ is replaced
in $P$ by inserting $B\cap S$ immediately followed by  $B\smallsetminus
S$. This guarantees that at the exit of 
$\mathit{Refine}(\tuple{P,\Rel},S)$ at line~7 of $\SSA$ 
 the list $P$ is still in reverse topological ordering w.r.t.\
$\Rel$. This is a consequence of the fact that at the exit of
$\mathit{Refine}(\tuple{P,\Rel},S)$, by point~(5) in
Section~\ref{asa}, we have that $\mu_{\tuple{P,\Rel}} (B\cap S) =
\mu_{\tuple{P,\Rel}}(B)\cap S$, i.e., $\mu_{\tuple{P,\Rel}} (B\cap S)
\cap (B\smallsetminus S) = \varnothing$ 
 so that $B\cap S \not\tle  B\smallsetminus S$.
The reverse topological ordering of $P$ w.r.t.\ $\tle$ ensures that 
if $(B,C)
\tl (B',C')$ then 
$(B,C)$ is scanned by $\mathit{FindRefiner}$ after the pair $(B',C')$. Since
$\mathit{FindRefiner}()$ exits as soon as a refiner is found, 
we have that $(B',C')$ cannot
be a refiner, so that condition $(*)$ holds  for $(B,C)$. 
\\
\indent
When $\mathit{FindRefiner}()$ determines that a pair of
 blocks $(B,C)$, with $B\raee C$, is not a refiner, it stores this
 information in a local boolean matrix Refiner that is indexed over
$P\times P$ and initialized to $\cmaybe$. Thus, the meaning of the matrix
Refiner is as follows: if
$\Refiner(B,C) = \cff$ then $(B,C)\not\in \Refiner(\cP)$.
If
$(B,C) \not \in \Refiner(\cP)$ then both (i) and (ii) do not hold, therefore
$\mathit{FindRefiner}()$ executes the for-loop at lines~13-14 so
that any $(B,E)$ with $E\tle C$ is marked as $\text{Refiner}(B,E) =
\cff$. This is correct because if $(B,C)\not\in \Refiner(\cP)$ and
$(B,E) \tle (B,C)$ then $(B,E)\not\in \Refiner(\cP)$: in fact, by
Lemma~\ref{l2}, $\Bottom(\mu_\cP(B)) \subseteq \mu_\cP(C) \cup
\pre(\mu_\cP (C))$, and since $E\tle C$ implies, because $\tle$ is
transitive, $\mu_\cP(C) \subseteq
\mu_\cP (E)$, we have that $\Bottom(\mu_\cP(B)) \subseteq \mu_\cP(E) \cup
\pre(\mu_\cP (E))$, so that, 
by Lemma~\ref{l2}, $(B,E)\not\in \Refiner(\cP)$. 
The for-loop at lines~13-14 is therefore an optimization of 
Theorem~\ref{th-refiner} since it determines that some
pairs of blocks are not a refiner without resorting to the 
condition $\neg\text{(i)}\wedge\neg\text{(ii)}$ of Theorem~\ref{th-refiner}. 
This optimization and the related matrix Refiner turn out to be
crucial for obtaining the overall time complexity of $\SSA$.

%Thus, if we consider $(B,C)$ such that $B\raee C$ that is 
%selected at a generic iteration of the for-loops at lines~6
%and 7  then  the
%test $\text{Refiner}(B,C)=\cmaybe$ at line~9  guarantees that $(B,C)$
%satisfies the condition $(*)$ required by Theorem~\ref{th-refiner}.  

\linesnumbered
\begin{algorithm}[Ht]
\small
\SetVline
\Setnlskip{1.5em}
\SetAlTitleFnt{textsc}
\Indm

\KwSty{Precondition: } $\TS(S,\sra,P_\ell)~~~\&~~~\forall x,y\in S.\ P_\ell(x)=P_\ell(y)$

\BlankLine
\clist{\cstate} $\pos$(\clist{\cstate} $S$, \clist{\cstate} $T$) \{

\Indp
\clist{\cstate} $R :=\varnothing$\;
\lForAll{$s \in S$}{mark1$(s)$\;}
\ForAll{$t \in T$}{
%  mark2$(t)$; $R$.append($t$)\;
  \ForAll{{\rm $s\in \pre(t)$ \KwSty{such that} $\text{marked1}(s)$}}{
    $\text{mark2}(s)$; $R$.append($s$)\;
  }
}
\ForAll{{\rm $y \in S$ \KwSty{backward} \KwSty{such that} $\text{marked2}(y)$}}{
  \ForAll{{\rm $x\in \pre(y)$ \KwSty{such that}
      $\text{marked1}(x)~\&~\text{unmarked2}(x)$}}{
      mark2$(x)$; $R$.append($x$)\;
  }
}
\lForAll{$x \in S$}{ unmark1$(x)$;} \lForAll{$x \in R$}{ unmark2$(x)$\;}
  
\creturn $R$\;

\Indm
\FuncSty{\}}

\caption{Computation of $\pos$.}\label{fig-eu}
\end{algorithm}

\subsection{Computing pos}\label{euc}
Given two lists of states $S$ and $T$, we want to compute the set of
states that belong to $\pos(S,T)$. This can be done by traversing
once the edges of
the transition relation $\sra$ provided that the list $\Sigma$ of states
satisfies the following property: 
\begin{center}
For all $x,y\in \Sigma$, 
if $x$ precedes $y$ in the list $\Sigma$ and $\ell (x)=\ell (y)$ 
then $y\not\!\sra x$. 
\end{center}
We denote this property by 
$\TS(\Sigma,\sra,P_\ell)$. 
Hence, this is a  topological ordering of $\Sigma$ w.r.t.\ the
transition relation $\sra$ that is local to each block of the labeling
partition $P_\ell$. 
%This can be obtained initially as follows. 
As described in
Section~\ref{sec:bs}, as an initial
pre-processing step of $\SSA$, we find and collapse inert s.s.c.'s. 
After this pre-processing step, $\Sigma$ is then
topologically ordered locally to each block of $P_\ell$ in $O(|\Sigma|
+ |\sra|)$ time in order to establish initially $\TS(\Sigma,\sra,P_\ell)$. 
%(cf.\ the call 
%$\mathit{TopologicalSort}(\Sigma,\sra,P_\ell)$ in the function 
%$\mathit{Initialize()}$ in
%Figure~\ref{fig-2} in Appendix~\ref{app}). 
We will see in Section~\ref{sec-sp} 
that while the ordering of the list $\Sigma$ of states
changes across the execution of $\SSA$, the property 
$\TS(\Sigma,\sra,P_\ell)$ is always maintained invariant.

The computation of $\pos(S,T)$ is done by the algorithm in
Figure~\ref{fig-eu}. The result $R$ consists of all the states in $S$
that are marked2. We assume that all the states in $S$ have the same
labeling by $\ell$: this is clearly true when the function $\pos$ is
called from the algorithm $\SSA$. 
The for-loop at lines~5-7 makes  the states in $S\cap
\pre(T)$ marked2.   Then, the for-loop at
lines~8-10 scans backward the list of states $S$ and when a
marked2 state $y$ is encountered then all the states in $S\cap
\pre(y)$ are marked2. It is clear that the property
$\TS(\Sigma,\sra,P_\ell)$ guarantees that this procedure does not
miss states that are in $\pos(S,T)$. 
%Thus, the function in
%Figure~\ref{fig-eu} allows us to compute $\pos(S,T)$ in $O(|\Sigma| +
%|\sra|)$ time.  

%\incmargin{1em}
\linesnumbered
\begin{algorithm}[Htp]
\small
\SetVline
\Setnlskip{1.5em}
\SetAlTitleFnt{textsc}
\Indm

\clist{\cblock} \FuncSty{$\mathit{Split}($\clist{\cblock} $P,$ \clist{\cstate} $S) \;\{$}

\Indp

\clist{\cblock} split\;
\ForAll{$x \in S$}{
  \If{{\rm $(x.\text{block}.\text{intersection} = \cnull)$}}{
    \cblock  ~$B$ $:=$ new \cblock\;
    %$B$.begin, $B$.end $:=$ $x$.block.begin\;
    $x$.block.intersection $:=$ $B$\;
    split.append($x$.block);
  }
  move $x$ in the list $\Sigma$ from $x$.block at the end of $B$\;
  \lIf{{\rm ($x$.block $= \varnothing$)}}{
   $x$.block $:= \text{copy}(B)$; $x$.block.intersection $:= \cnull$\;
 }
}
\ForAll{{\rm $B\in \text{split}$}}{
  \lIf{{\rm $(B$.intersection $= \cnull)$}}{
    split.remove($B$); delete $B$\;
  }
  \lElse{insert $B$.intersection in $P$ in front of $B$\;}
}

\creturn split\;

\Indm
\FuncSty{\}}

\BlankLine
\BlankLine

\cvoid \FuncSty{$\mathit{SplittingProcedure}(\PR\;\tuple{P,\Rel},$
  \clist{\cstate} $S)\;\{$}

\Indp
  \clist{\cblock} split $:=\mathit{Split(P,S)}$\;
  \If{{\rm $(\text{split} \neq \varnothing)$}}{
    resize $\Rel$; \tcp{{\rm update $\Rel$}}
    \lForAll{$B \in P$}{
      \lForAll{{\rm $C \in \text{split}$}}{$\Rel(C.\text{intersection},B) := \Rel(C,B)$\;}
    }
    \lForAll{{\rm $B \in \text{split}$}}{ 
      \lForAll{$C \in P$}{$\Rel(C,B.\text{intersection}) := \Rel(C,B)$\;}
    }
    $\mathit{Update}()$; \tcp{{\rm update Count, BCount, localBottoms, bottomBlocks}}
    \lForAll{$B \in P$}{$B.\text{intersection} := \cnull$\;}
  }
  \Indm
  \FuncSty{\}}

\caption{Splitting Procedure.}\label{fig-sp}
\end{algorithm}

\subsection{$\mathit{SplittingProcedure}$}\label{sec-sp}

$\SSA$ calls $\mathit{SplittingProcedure}(\tuple{P,\Rel},S)$ at line~6 
with the precondition $\TS(\Sigma,\sra,P_\ell)$ and  
needs to maintain this invariant property at the exit (as discussed
in Section~\ref{euc} this is crucial for computing $\pos$).  
This function must modify the current PR $\cP= \tuple{P,\Rel}$ to 
$\cP' = \tuple{P',\Rel'}$ as follows: 
\begin{itemize}
\item[{\rm (A)}] 
$P'$ is the partition obtained by splitting $P$ w.r.t.\ the splitter
$S$; 
\item[{\rm (B)}] $\Rel$ is modified to $\Rel'$ in 
such a way that for any $x\in \Sigma$, $\mu_{\cP'}(P'(x)) =
\mu_{\cP}(P(x))$. 
\end{itemize}
Recall that the states of a block $B$ of 
  $P$ are consecutive in the list $\Sigma$, so that $B$ is represented
  as
  $B=[B.\text{begin},B.\text{end}[$.
An implementation of the splitting operation 
$\mathit{Split}(P,S)$ that only
scans the states in $S$, i.e.\ that takes $O(|S|)$
time, is quite easy and standard (see e.g.\
\cite{gv90,rt07}). However, this operation affects the ordering of the
states in the list 
$\Sigma$ because states are moved from old blocks to newly generated
blocks. It turns out that this splitting operation 
can be implemented in a careful way that preserves  the invariant property
$\TS(\Sigma,\sra, P_\ell)$. The idea is rather simple. Observe that 
the list of states $S = \pos(\mu_\cP(X), \mu_\cP(Y))$ 
can be (and actually is) built as a sublist of $\Sigma$ so that the
following property holds: If $x$
precedes $y$ in $S$ and $P_\ell(x)=P_\ell(y)$ then $y\!\not\!\sra x$.  
The following picture shows the idea of our implementation of
$\mathit{Split}(P,S)$, where states within filled circles determine the
splitter set $S$.

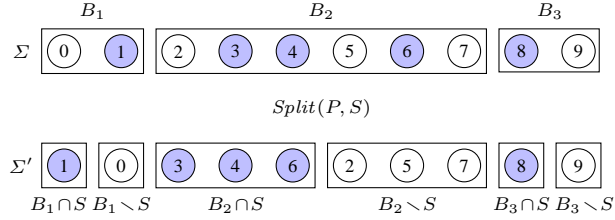
\begin{center}
    \begin{tikzpicture}[scale=0.76]
      \scriptsize
      \tikzstyle{arrow}=[->,>=latex']
      \draw (-0.75,2) node {$\Sigma$};
      \draw (-0.75,0) node {$\Sigma'$};
      \draw (0.5,2.7) node {$B_1$};
      \draw (4.5,2.7)  node {$B_2$};
      \draw (8.5,2.7)  node {$B_3$};
      \draw (4.5,1)  node {$\mathit{Split}(P,S)$};
      \draw (-0.1,-0.7) node {$B_1\!\cap\! S$};
      \draw (1,-0.7)  node {$B_1\!\smallsetminus\! S$};
      \draw (3,-0.7)  node {$B_2\!\cap\! S$};
      \draw (6,-0.7)  node {$B_2\!\smallsetminus\! S$};
      \draw (8,-0.7)  node {$B_3\!\cap\! S$};
      \draw (9.1,-0.7)  node {$B_3\!\smallsetminus\! S$};

      \path 
      (0,2) node[shape=circle,draw,name=0]{0}
      (1,2) node[shape=circle,fill=blue!25,draw,name=1]{1}
      (2,2) node[shape=circle,draw,name=2]{2}
      (3,2) node[shape=circle,fill=blue!25,draw,name=3]{3}
      (4,2) node[shape=circle,fill=blue!25,draw,name=4]{4}
      (5,2) node[shape=circle,draw,name=5]{5}
      (6,2) node[shape=circle,fill=blue!25,draw,name=6]{6}
      (7,2) node[shape=circle,draw,name=7]{7}
      (8,2) node[shape=circle,fill=blue!25,draw,name=8]{8}
      (9,2) node[shape=circle,draw,name=9]{9};
      % \draw[arrow] (0) to (1);
%       \draw[arrow] (1) to (2);
%       \draw[arrow] (2) to (3);
%       \draw[arrow] (3) to (4);
%       \draw[arrow] (4) to (5);
%       \draw[arrow] (5) to (6);
%       \draw[arrow] (6) to (7);
%       \draw[arrow] (7) to (8);
%       \draw[arrow] (8) to (9);
      \path 
      (0,0) node[shape=circle,fill=blue!25,draw,name=00]{1}
      (1,0) node[shape=circle,draw,name=11]{0}
      (2,0) node[shape=circle,fill=blue!25,draw,name=22]{3}
      (3,0) node[shape=circle,fill=blue!25,draw,name=33]{4}
      (4,0) node[shape=circle,fill=blue!25,draw,name=44]{6}
      (5,0) node[shape=circle,draw,name=55]{2}
      (6,0) node[shape=circle,draw,name=66]{5}
      (7,0) node[shape=circle,draw,name=77]{7}
      (8,0) node[shape=circle,fill=blue!25,draw,name=88]{8}
      (9,0) node[shape=circle,draw,name=99]{9};

      \path (0.north west) ++(-0.18,0.18) node[name=b1a]{} (1.south
      east) ++(0.18,-0.18) node[name=b1b]{};
      \path (2.north west) ++(-0.18,0.18) node[name=b2a]{} (7.south
      east) ++(0.18,-0.18) node[name=b2b]{};
      \path (8.north west) ++(-0.18,0.18) node[name=b3a]{} (9.south east) ++(0.15,-0.15) node[name=b3b]{};

     \draw[rounded corners=0pt] (b1a) rectangle (b1b);
     \draw[rounded corners=0pt] (b2a) rectangle (b2b);
     \draw[rounded corners=0pt] (b3a) rectangle (b3b);

      \path (00.north west) ++(-0.18,0.18) node[name=c1a]{} (00.south
      east) ++(0.18,-0.18) node[name=c1b]{};
      \path (11.north west) ++(-0.18,0.18) node[name=c2a]{} (11.south
      east) ++(0.18,-0.18) node[name=c2b]{};
      \path (22.north west) ++(-0.18,0.18) node[name=c3a]{} (44.south
      east) ++(0.18,-0.18) node[name=c3b]{};
      \path (55.north west) ++(-0.18,0.18) node[name=c4a]{} (77.south
      east) ++(0.18,-0.18) node[name=c4b]{};
      \path (88.north west) ++(-0.18,0.18) node[name=c5a]{} (88.south
      east) ++(0.18,-0.18) node[name=c5b]{};
      \path (99.north west) ++(-0.18,0.18) node[name=c6a]{} (99.south
      east) ++(0.18,-0.18) node[name=c6b]{};

     \draw[rounded corners=0pt] (c1a) rectangle (c1b);
     \draw[rounded corners=0pt] (c2a) rectangle (c2b);
     \draw[rounded corners=0pt] (c3a) rectangle (c3b);
     \draw[rounded corners=0pt] (c4a) rectangle (c4b);
     \draw[rounded corners=0pt] (c5a) rectangle (c5b);
     \draw[rounded corners=0pt] (c6a) rectangle (c6b);

  \end{tikzpicture}    
\end{center}
The property $\TS(\Sigma',\sra,P_\ell)$ still holds for the modified list of
states $\Sigma'$. In fact, from the above picture 
observe that it is enough to check that: if $B$ has been split into
$B\cap S$ and $B\smallsetminus S$ by preserving the relative orders of
the states in $\Sigma$ then if $x\in B\cap S$ and $y\in
B\smallsetminus S$ then $y\!\not\!\sra x$. This is true because if
$y\sra x$ and $x
\in S =\pos(\mu_\cP(X), \mu_\cP(Y))$ then, since $x$ and $y$ are in
the same block of $P$ and $\mu_\cP(X)$ is a union of some blocks of
$P$,  by definition of $\pos$ we would also have that
$y\in S$, which is a contradiction.

The functions in Figure~\ref{fig-sp} sketch a pseudo-code that implements
the above described splitting operation (the $\mathit{Update}()$
function is in
Figure~\ref{fig-bis} in Appendix~\ref{app}). The above point~(B), i.e.,
the modification of $\Rel$ to $\Rel'$ so that 
for any $x\in \Sigma$, $\mu_{\cP'}(P'(x)) =
\mu_{\cP}(P(x))$ is straightforward and is implemented at lines~18-20
of  $\mathit{SplittingProcedure}()$. 

% Finally, let us remark that for all $B\in P$:

% \medskip
% \noindent (1)~$B=\parent(B)\smallsetminus S$;\\
% \smallskip
% (2)~$B\neq \parent(B)
% \Ra B.\text{intersection}=\parent(B)\cap S$;\\ 
% \smallskip
% (3)~$B=\parent(B)
% \Ra B.\text{intersection}=\cnull$.
% \\
% \smallskip
% \noindent
% (4)~$\text{split}=\{C\in P~|~C\neq \parent(C), 
% C=\parent(C)\!\smallsetminus \! S\}$.

%\decmargin{0.3em}
\linesnumbered
\begin{algorithm}[htp]
\small
\SetVline
\Setnlskip{1.5em}
\SetAlTitleFnt{textsc}
\Indm

 \cvoid \FuncSty{$\mathit{Refine}(\PR\;\tuple{P,\Rel},$
  \clist{\cstate} $S)\;\{$}

\Indp
\clist{\cblock} $L := \varnothing$\;
\lForAll{{\rm $s \in S$ \KwSty{such that} $\text{unmarked}(s.\text{block})$}}{ 
    $\text{mark}(s.\text{block})$;
    $L.\text{append}(s.\text{block})$\;
  }
\ForAll{$B \in L$}{
  \ForAll{$C \in P$}{
    \If{{\rm $(\Rel(B,C) = \ctt ~\&~ \text{unmarked}(C))$}}{
      $\Rel(B,C) := \cff$\;
      %\tcp{{\rm update Count and BCount}}
 
      \ForAll{$y\in C$}{
        \lForAll{$x\in \pre(y)$}{
          $\text{Count}(x,B)\,$-- --;
          $\text{BCount}(x.\text{block},B)\,$-- --\;
        }
      }
      %\tcp{{\rm update $B$.bottomBlocks}}

      \lIf{{\rm $(C\in B.\text{bottomBlocks})$}}{$B$.bottomBlocks.erase($C$)\;}
      \ForAll{$y\in C$}{
        \ForAll{{\rm $x\in \pre(y)$}}{
          \If{{\rm $(x.\text{block} \neq B \,\& \Rel(B,x.\text{block})\!= \ctt
              \,\&\, \text{Count}(x,B)=0)$}}{
            $\text{mark2}(x.\text{block})$\; 
            \If{{\rm \text{unmarked2}(x.\text{block})}}{$B$.bottomBlocks.append($x$.block)\;}
          }
        }
      }
    }
  }
}
  \lForAll{$B \in P$}{$\text{unmark}(B)$; $\text{unmark2}(B)$}\;
\Indm
\}

\caption{$\mathit{Refine}$ function.}\label{fig-ref}
\end{algorithm}

\subsection{$\mathit{Refine}$ Function} 

$\SSA$ calls $\mathit{Refine}(\tuple{P,\Rel},S)$ at line~7 with the
precondition that $S$ is a union of  blocks of the current partition
$P$. 
The function $\mathit{Refine}(\tuple{P,\Rel},S)$ in
Figure~\ref{fig-ref} implements the point~(5) of
Section~\ref{asa}. This function must modify the current PR $\cP=\tuple{P,\Rel}$ to
$\cP'=\tuple{P,\Rel'}$ by pruning 
the relation $\Rel$ in such a way that for any $B\in P$:
 $$\mu_{\cP'}(B) = \left\{\begin{array}{ll}
                      \mu_{\cP}(B)\cap S    & \mbox{~~if $B\subseteq S$}\\
                      \mu_{\cP}(B) & \mbox{~~if $B\cap S=\varnothing$}
                      \end{array}\right.
$$
This is done by the $\mathit{Refine}()$ function at lines~5-7
by reducing the relation $\Rel$ to $\Rel'$ as follows: if 
$B,C\in P$ and $\Rel (B,C)=\ctt$ then $\Rel' (B,C)=\cff$
iff  $B\subseteq S$ and $C\cap S = \varnothing$,
while the rest of the code updates the data
structures Count, BCount and bottomBlocks accordingly (note that
localBottoms do not need to be updated).

\subsection{Auxiliary Functions}

It is straightforward to implement the remaining 
functions $\mathit{Initialize}()$ and $\mathit{Image}()$ (these are
given in Figure~\ref{fig-2} in Appendix~\ref{app}). 
It is just worth remarking that in $\mathit{Initialize}()$, 
$\mathit{TopologicalSort}(\Sigma,\sra,P)$ establishes initially the property 
$\TS(\Sigma,\sra,P_\ell)$, while
the call $\mathit{TopologicalSort}(P,\Rel)$ provides an initial reverse topological
order of $P$ w.r.t.\ $\Rel$ when the input partial ordering
$\Rel$ is not the
identity relation $\id$.

\subsection{Complexity}

% The correcteness of our stuttering simulation algorithm $\SSA$ is
% proved by the following result together with
% the description in Section~\ref{sec-impl} 
% of the correctness of the functions called by 
% $\BasicSSA$.

Time and space bounds for $\SSA$ are as follows. In the following 
statement we assume, as
usual in model checking, that the transition relation $\sra$ is
total, i.e., for any $s\in \Sigma$ there exists
$t\in \Sigma$ such that $s\sra t$, so that the inequalities 
$|\Sigma|\leq |\sra|$ and
$|\Pstsim| \leq |\sra^{\exists}|$ hold and this 
allows us to simplify the expression of the time
bound.  

\begin{theorem}[\textbf{Complexity}]\label{th-complexity}
$\SSA$ runs in 
$O(|\Pstsim|^2(|\sra|+|\Pstsim||\sra^{\exists}|))$-time and
$O(|\Sigma||P_{\mathrm{stsim}}| \log |\Sigma|)$-space.  
\end{theorem}

\subsection{Adapting SSA for LTSs}
The algorithms $\SSA$ computes the stuttering simulation preorder on
KSs, but it can be modified to work over LTSs by
following the adaptation to LTSs of Groote and Vaandrager's
algorithm~\cite{gv90} for KSs. Due to lack of space the
details are here omitted. We just mention that for
any action $a\in \mathit{Act}$, we have a parametric $\pos_a$
operator for any action $a\in \mathit{Act}$ 
so that the notions of splitting and refinement 
of the current PR are parameterized w.r.t.\ the action $a$.

\section{Conclusion}
We presented an algorithm, called $\SSA$, for computing the stuttering
simulation preorder and equivalence on a Kripke structure or labeled
transition system.  To the
best of our knowledge, this is the first algorithm for computing this
behavioural preorder.  The only available algorithm related to
stuttering simulation is a procedure by Bulychev et al.~\cite{bkz07}
that checks whether a given relation is a stuttering simulation. Our
procedure $\SSA$ includes an algorithm for checking whether a given
relation is a stuttering simulation that significantly
improves Bulychev et al.'s one both in time and in space.

%\medskip
%\noindent
\paragraph{\textbf{Acknowledgements.}} 
%We are grateful to Silvia Crafa for her kind helpful discussions. 
This work was partially supported 
%by  the FIRB Project ``Abstract interpretation 
%and model checking for the verification of embedded systems'', 
by the PRIN 2007 Project ``\emph{AIDA\-2007: Abstract
Interpretation Design and Applications}'' and by the University of
Padova under the Projects ``\emph{Formal
methods for specifying and verifying behavioural properties of
software systems}'' and ``\emph{Analysis, verification and abstract 
interpretation of models for concurrency}''.

%\newpage
%\mbox{}
\newpage
\appendix 

\section{Appendix}\label{app}

\begin{lemma}
  \label{Sitrans}
At the beginning of any iteration of $\BasicSSA$,  $\StSim$ is a
preorder.  
\end{lemma}

\begin{proof}
  Initially, $\StSim$ is reflexive and transitive because $\{\StSim(x)\}_{x\in
    \Sigma}$ is a partition. Let us denote by
    $\StSim_i$ the value of $\StSim$ 
at the beginning of the $i$-th iteration of $\BasicSSA$. Then, 
 $$\StSim_{i+1}(x) = \left\{\begin{array}{ll}
                      \StSim_i(x)\cap S    & \mbox{~~if $x\in S$}\\
                      \StSim_i(x) & \mbox{~~if $x\not\in S$}
                      \end{array}\right.
$$
Then, by inductive hypothesis, $\StSim$ is clearly reflexive. Let us
turn on transitivity. 
Consider $z\in \StSim_{i+1}(y)\subseteq \StSim_i(y)$ 
and $y\in \StSim_{i+1}(x)\subseteq \StSim_i(x)$. Then,
by inductive hypothesis, $z\in \StSim_i (x)$. If $x\not \in S$ then 
$\StSim_i (x)=\StSim_{i+1}(x)$ and therefore $z\in \StSim_{i+1} (x)$. 
If, instead, $x\in S$ then $\StSim_{i+1}(x)=\StSim_i(x)\cap S$ and therefore
$y\in S$. Hence, $\StSim_{i+1}(y) = \StSim_i(y) \cap S$ so that $z\in
S$, i.e.\ $z\in \StSim_i(x)\cap S = \StSim_{i+1}(x)$. \qed
\end{proof}

\bigskip
\noindent
\emph{Proof of Lemma~\ref{basicssa}.}
The output relation 
$\StSim$ is a stuttering simulation so
  that $\StSim\subseteq\Rstsim$.
  Thus, we need to prove that $\StSim\supseteq\Rstsim$.  
Let us denote by
    $\StSim_i$ the value of $\StSim$ 
at the beginning of the $i$-th iteration of $\BasicSSA$.
We show by induction on $i$ that $\Rstsim \subseteq 
\StSim_i$. 

\medskip
\noindent
$(i=0)$ $\Rstsim \subseteq \StSim_0$ because $\StSim_0 (x)=P_\ell
(x)$. 

\medskip
\noindent
$(i+1)$ Let us prove that for any $w$,
    $\Rstsim (w) \subseteq \StSim_{i+1}(w)$, where
 $$\StSim_{i+1}(w) = \left\{\begin{array}{ll}
                      \StSim_i(w)\cap S    & \mbox{~~if $w\in S$}\\
                      \StSim_i(w) & \mbox{~~if $w\not\in S$}
                      \end{array}\right.
$$
$S=\pos(\StSim_i(x),\StSim_i(y))$, $x\sra y$ and
$\StSim_i(x)\not\subseteq S$. 

\noindent
If $w\not \in S$ then $\StSim_{i+1}(w)=\StSim_i(w) \supseteq
\Rstsim(w)$. If, instead, $w\in S$ then
$\StSim_{i+1}(w)=\StSim_i(w)\cap S$. By inductive hypothesis,
$\StSim_i(w) \supseteq \Rstsim(w)$, therefore it is enough to show
that $S=\pos(\StSim_i(x),\StSim_i(y)) \supseteq \Rstsim (w)$. Consider
$v\in \Rstsim(w)$. Since $w\in \pos(\StSim_i(x),\StSim_i(y))$, there
exists a path $w=u_0\sra u_1\sra\ldots u_{n-1}\sra u_n$ such that for
any $j\in [0,n)$, $u_j\in \StSim_i(x)$ and $u_n\in \StSim_i(y)$.  It
turns out that any transition $u_j \sra u_{j+1}$ can be lifted to a
path $$w^j_{0}\sra\ldots\sra w^j_{m_j -1}\sra w^j_{m_j}$$ where
$w^j_k\in\Rstsim(u_j)$ when $k\in [0,m_j)$ and
$w^j_{m_j}\in\Rstsim(u_{j+1})$, and in particular $w^0_{0}=v$. 
In fact, consider the first transition $w=u_0 \sra u_1$. Since $v\in
\Rstsim(w)$, there exists $w^0_0,...,w^0_{m_0}$ such that $v=w^0_0 \sra w^0_1 \sra
... \sra w^0_{m_0}$ where $w^0_l \in \Rstsim(u_0)$ for any $l\in [0,m_0)$ and
$w^0_{m_0}\in \Rstsim(u_1)$. Thus, by a simple induction, any transition
$u_j \sra u_{j+1}$ can be
lifted to one such path. Moreover, by induction, for any $j\in [0,n)$, 
$\Rstsim(u_j) \subseteq \StSim_i (u_j)$, while $\Rstsim(u_n) \subseteq
\StSim_i (u_n)$. By Lemma~\ref{Sitrans}, $\StSim_i$ is transitive so
that from $\{u_0,\ldots,u_{n-1}\}\subseteq \StSim_i(x)$ and $u_n\in
\StSim_i(y)$ we obtain that for any $j\in [0,n)$,
$\StSim_i(u_j)\subseteq \StSim_i(x)$ and $\StSim_i(u_n)\subseteq
\StSim_i(y)$. 
The concatenation of the above paths therefore provides a path 
$$v=w_0 \sra w_1 \sra \ldots \sra w_{n-1} \sra w_n$$ 
such that for any $l\in [0,n)$, $w_l \in \StSim_i(x)$ and $w_n\in
\StSim_i(y)$.  
    Consequently, $v\in\pos(\StSim_i(x),\StSim_i(y))$ and this
    concludes the proof.  \qed

% \begin{lemma}\label{st-char}
% $\cP= \tuple{P,\tle}$ is a stuttering simulation iff $P\preceq
% P_\ell$ and for all $B,C\in P$, 
% $\pos(\mu_\cP(B),\mu_\cP(C)) =
% \mu_\cP(\pos(\mu_\cP(B),\mu_\cP(C)))$. 
% \end{lemma}
% \begin{proof}
% Immediate from Lemma~\ref{stutterEU} because
% $\pos(\mu_\cP(B),\mu_\cP(C))=\mu_\cP(\pos(\mu_\cP(B),\mu_\cP(C)))$
% iff for any $x\in\pos(\mu_\cP(B),\mu_\cP(C))$,
% $\mu_\cP(x)\subseteq\pos(\mu_\cP(B),\mu_\cP(C))$.
% \qed
% \end{proof}

\bigskip
\noindent
\emph{Proof of Theorem~\ref{st-char2}.}
$(\Ra)$ If $\mu_\cP$ is a stuttering simulation and $t\in \mu_\cP(s)$
then $\ell(t)=\ell(s)$, i.e., $t\in P_\ell(s)$. 
Moreover,  if $B\raee C$ then there exists $s\in B$ and $s'\in C$
such that $s\sra s'$, so that
$\mu_\cP(s) \subseteq \pos(\mu_\cP(s),\mu_\cP (s'))$. Since
$\mu_\cP(s)=\mu_\cP(B)$ and $\mu_\cP(s')=\mu_\cP(C)$, we have that 
$\mu_\cP(B) \subseteq \pos(\mu_\cP(B),\mu_\cP (C))$. Hence, 
$\refiner(\cP)=\emptyset$. 

\noindent
$(\La)$ Assume that $s\sra s'$ and $t\in \mu_\cP (s)$. Therefore,
$t\in P_\ell(s)$, i.e., $\ell(t)=\ell(s)$. Furthermore,  
$P(s)\raee P(s')$, so that from $\refiner(\cP)=\emptyset$ we obtain  that
$\mu_\cP(P(s)) \subseteq \pos(\mu_\cP (P(s)), \mu_\cP(P(s')))$. 
Since $\mu_\cP(P(s)) = \mu_\cP(s)$ and $\mu_\cP(P(s')) = \mu_\cP(s')$, 
we have that 
$\mu_\cP(s) \subseteq \pos(\mu_\cP (s), \mu_\cP(s'))$, and therefore
$\mu_\cP$ is a stuttering simulation.
\qed

\medskip
\noindent
\emph{Proof of Theorem~\ref{correctness}.}
This is a consequence of the following two facts.
Let $\StSim$ be the current relation in $\BasicSSA$ at the end of some
iteration  and let $\cP=\tuple{P,\tle}$ be the corresponding PR.
\begin{itemize}
\item[{\rm (i)}]  
We have that
$(x,y)\in \Sigma^2$ is a refiner in $\BasicSSA$ iff
  $(P(x),P(y))\in P^2$ is a refiner in $\SSA$. This is true because for any
  $x,y\in \Sigma$, we have that $\StSim(x)\not\subseteq
  \pos(\StSim(x),\StSim(y))$ iff $\mu_\cP(P(x))\not\subseteq
  \pos(\mu_\cP(P(x)),\mu_\cP(P(y)))$.   
\item[{\rm (ii)}] Let $(x,y)\in \Sigma^2$ be a refiner in
  $\BasicSSA$ and
  $S=\pos(\StSim(x),\StSim(y))=\pos(\mu_\cP(P(x)),\mu_\cP(P(x)))$. 
Let $P'=\Split(P,S)$. 
Consider 
$$\StSim'(x) = \left\{\begin{array}{ll}
                      \StSim (x)\cap S    & \mbox{~~if $x\in S$}\\
                      \StSim(x) & \mbox{~~if $x\not\in S$}
                      \end{array}\right.
$$
$$
\mu_{\cP}'(B) = \left\{\begin{array}{ll}
                      \mu_{\cP} (B)\cap S    & \mbox{~~if $B\subseteq S$}\\
                      \mu_\cP(B) & \mbox{~~if $B\cap S=\varnothing$}
                      \end{array}\right.
$$
where $x\in \Sigma$ and $B\in P'$. Then, for any $x\in \Sigma$,
$\StSim'(x) = \mu_{\cP}'(P'(x))$. \qed
\end{itemize}

\bigskip
\noindent
\emph{Proof of Lemma~\ref{l2}.}
Let $\mu = \mu_\cP$ and $(B,C)\in P^2$ such that
$B\raee C$. 
\\
\indent
Assume that $\mu(B) \subseteq \pos(\mu(B),\mu(C))$ and consider 
  $b\in\bottom(\mu(B))$. 
  Then, $b\in \pos(\mu(B),\mu(C))$, so that there exist $x_0,...,x_k\in
  \mu(B)$, with $k\geq 0$, such
  that $b=x_0$, for all $i\in [0,k)$, $x_i\in 
\mu(B)$ and $x_i \sra x_{i+1}$, and $x_k\in \mu(C)$. 
 Since $b\in
  \bottom(\mu(B))$, we have that $b\not\in \pre(\mu(B))$ and therefore
  necessarily
  either $k=0$ or $k=1$. If $k=0$ then 
$b\in \mu(B)\cap
  \mu(C)$. If instead $k=1$ then $b\in \pre(\mu(C))$. Thus, $b\in
  \mu(C) \cup \pre(\mu(C))$.
\\
\indent
  Conversely, assume that
  $\bottom(\mu(B))\subseteq\mu(C)\cup\pre(\mu(C))$ and consider $x\in
  \mu(B)$. If $x\in \bottom(\mu(B))$ then clearly $x\in \pos(\mu(B),\mu(C))$.  
If instead $x\not \in \bottom(\mu(B))$ then $x\in \pre(\mu(B))$, so that  
  there  exists $y\in\mu(B)$ such that $x\sra y$. Again, if
  $y\in\bottom(\mu(B))$ then $y\in\pos(\mu(B),\mu(C))$ and therefore
  we have that
  $x\in\pos(\mu(B),\mu(C))$.  If $y\not\in \bottom(\mu(B))$ then 
  we can go on with this construction.
  Since $\Sigma$ is finite, in this way 
we would obtain a cycle of inert transitions inside $\mu(B)\subseteq
P_\ell (B)$, namely a contradiction. Thus, it must 
exist some $z\in \Bottom(\mu(B))$  such that $x\sra^* z$, and
therefore $x\in\pos(\mu(B),\mu(C))$. Hence, 
  $\mu(B) \subseteq \pos(\mu(B),\mu(C))$. \qed

\bigskip
\noindent
\emph{Proof of Theorem~\ref{th-refiner}.}
Let us first observe that since $\trianglelefteq$ is a preorder, and
therefore transitive, if
$B\trianglelefteq C$ then $\mu_\cP (C)\subseteq \mu_\cP (B)$. \\
  $(\Ra)$ Let us assume that 
  $(B,C)\not\in\Refiner(\cP)$. 	
   If $C\tle B$ then
  both conditions~(i) and (ii) trivially do not hold: for (ii), $D\in 
  \bottomBlock(B)$ implies $C\tle B \tl D$, and therefore $C\tl D$,
 which is in contradiction with
  $C\not\tle D$.
  Thus, assume that $C\not\trianglelefteq B$. Since
  $B\raee C$, by Lemma~\ref{l2}, we have that
  $\Bottom(\mu_\cP (B))\subseteq\mu_\cP (C)\cup\pre(\mu_\cP (C))$. Hence,
  $\Bottom(\mu_\cP (B))\cap B\subseteq(\mu_\cP (C)\cap B)\cup(\pre(\mu_\cP (C))\cap B)=\pre(\mu_\cP (C))\cap B \subseteq \pre(\mu_\cP (C))$, 
  because $C\not\trianglelefteq B$ implies 
  $B\cap\mu_\cP (C)=\varnothing$. 
  Moreover, if $C\not\trianglelefteq D$, then, again by
  Lemma~\ref{l2}, $\bottom(\mu_\cP (B))\cap D\subseteq \pre(\mu(C))\cap D$. 
  If $D\in \bottomBlock(B)$ then $\bottom(\mu_\cP (B))\cap D \neq \varnothing$
  and therefore $\pre(\mu(C))\cap D \neq \varnothing$, i.e., $D\raee \mu_\cP(C)$.

\medskip
\noindent
  $(\La)$ We prove that if (i) and (ii) do not hold then
 $(B,C)\not\in\Refiner(\cP)$. 	 
By Lemma~\ref{l2}, let us show that
  $\bottom(\mu_\cP (B))\subseteq\mu_\cP(C)\cup\pre(\mu_\cP(C))$.  If $C\trianglelefteq B$ then this is trivially true. 
  Thus, let us assume that $C\not\trianglelefteq B$. 
  \begin{align*}
    \bottom(\mu_\cP (B))
    &= \\
    \text{\small[as $\mu_\cP (B) = \cup \{D \in P~|~ B\tle D\}$]\hspace*{20ex}}
    &\\
     \cup \{D\cap\bottom(\mu_\cP (B))~|~ B\tle D\}
    & = \\
    \text{\small[by set theory]\hspace*{20ex}}
    &\\
    (B\cap\bottom(\mu_\cP(B)))\cup & \\
    \cup \{ D\cap \bottom(\mu_\cP (B))~|~ B \tl D, 
    C\not\tle D \} \cup &\\
    \cup \{ D\cap \bottom(\mu_\cP (B))~|~ B \tl D, 
    C\tle D \}
     &=\\
     \text{\small[by definition of $\bottomBlock$]\hspace*{20ex}} &\\   
     (B\cap\bottom(\mu_\cP(B)))\cup &\\ 
     \cup \{ D\cap \bottom(\mu_\cP (B))~|~D \in \bottomBlock(B), C
     \not\tle D\} \cup&\\ 
    \cup \{ D\cap \bottom(\mu_\cP (B)) ~|~ B \tl D, C\tle D \}
     &\subseteq \\
     \text{\small[by conditions~(i) and (ii)]\hspace*{20ex}} &\\   
	\pre(\mu_\cP (C)) \cup & \\
	\cup \{ D\cap \bottom(\mu_\cP (B)) ~|~ D \in
        \bottomBlock(B),C\not\tle D, D \raee \mu_\cP (C) \} \cup &\\
    \cup \{ D\cap \bottom(\mu_\cP (B)) ~|~ B \tl D, 
    C\tle D \} & 
    \subseteq \\
    \text{\small[because $C\tle D \;\Ra\; D\subseteq \mu_\cP (C)$]\hspace*{20ex}} &\\
    \pre(\mu_\cP (C)) \cup & \\
	\cup \{ D\cap \bottom(\mu_\cP (B)) ~|~ D \in \bottomBlock(B),C\not\tle D, D \raee \mu_\cP (C) \} \cup&\\ 
    \mu_\cP (C). & \\
    \end{align*}
\noindent 
Consider now $D\in \bottomBlock(\cP)$, $C\not\tle D$ and $D\raee \mu_\cP (C)$. 
Then, 
\begin{align*}
D\cap \bottom(\mu_\cP (B)) & = \\
D \cap \mu_\cP (B) \cap \neg \pre(\mu_\cP (B)) &=  \text{~~~\small[as 
$D\subseteq \mu_\cP (B)$]}\\
D \cap \neg \pre(\mu_\cP (B)) &\subseteq  \text{~~~\small[as $D\subseteq \mu_\cP (D) \subseteq \mu_\cP (B)$]}\\
\mu_\cP (D) \cap \neg \pre(\mu_\cP (D)) &= \\
\bottom(\mu_\cP (D)) &. \\
\end{align*}
Since $D\raee \mu_\cP (C)$, there exists $C\tle E$ such that $D\raee E$. 
Since $(B,C) \tle (D,E)$ and $D\raee E$, by hypothesis, 
$(D,E) \not\in \Refiner (\cP)$, so that, by Lemma~\ref{l2}, we have that
$\bottom(\mu_\cP (D)) \subseteq \mu_\cP (E) \cup 
\pre(\mu_\cP (E)) \subseteq \mu_\cP (C) \cup \pre(\mu_\cP (C))$. 
Thus, summing up, it turns out that 
$ \bottom(\mu_\cP (B)) \subseteq \mu_\cP (C) \cup \pre(\mu_\cP (C))$, so that,
by Lemma~\ref{l2}, $(B,C)\not\in \refiner(\cP)$. \qed

\bigskip
\noindent
\emph{Proof of Theorem~\ref{th-complexity}.}
\textbf{Time Complexity.} 
The time complexities of the various functions that are called by
$\SSA$ are as follows. 
\begin{itemize}
\item[--] $\mathit{Initialize()}$ takes $O(|P||\sra|)$ time. 
\item[--] $\mathit{FindRefiner()}$ takes 
$O(|P|^2 + |\sra|+|P||\sra^{\exists}|)$ time. This bound is computed
as follows. 
Line~4 takes $O(|P|^2)$ time. Lines~5-6 take $O(|\sra| + |P|^2)$ time.
Note that lines~5-6 are actually implemented as follows: 
%
%\decmargin{0.7em}
\vspace*{-20pt}
\linesnotnumbered
\restylealgo{plain}
\begin{algorithm}
\small
%\SetAlTitleFnt{textsc}
\ForAll{$C\in P$}{
  \lForAll{$y\in C$}{
    \lForAll{$x\in \pre(y)$}{
      $\text{mark}(x.block)$\;
    }
    }
    \ForAll{{\rm $B\in P$ \KwSty{such that} $\text{marked}(B)$}}{
\tcp{{\rm main body of $\mathit{FindRefiner()}$}}      
}
\lForAll{$B\in P$}{$\text{unmark}(B)$\;}
}
\end{algorithm}
\vspace*{-20pt}

Lines~11-12 and 13-14 take
 $O(|P||\sra^{\exists}|)$ time. 
 The estimate of the 
overall cost of lines~7-10
deserves special care. At line~10, it turns out that 
$\text{Count}(s,C)>0 \Lra s\raee C$:  if $\text{Count}(s,C)>0$ at line 10 then
$s\raee \cup \{E\in P~|~ E\tle C\}$. However, 
as a consequence of the code at lines~13-14, it turns out that when we are
at line~10, namely when Refiner$(B,C)=\cmaybe$, it is true that 
$\{E\in P~|~ E\tle C\}= \{C\}$ so that $s\raee C$. 
Hence, the overall cost of lines~7-10 is $\sum_{C\in P}\sum_{B\in P}
|\{(x,y)~|~x\in B, y\in C, x\sra y\}| \leq |\sra|$. 

\item[--] $\mathit{Image}(\tuple{P,\Rel},B)$ takes $O(|\Sigma|)$ time. 
\item[--] $\pos(S,T)$ takes $O(|\Sigma| + |\sra|)$ time. 

\item[--] $\mathit{SplittingProcedure}(\tuple{P,\Rel},S)$ takes 
$O(|P||\Sigma|)$ time. In particular, $\mathit{Split}(P,S)$
takes $O(|S|)$ time. 

\item[--] $\mathit{Refine}(\tuple{P,\Rel},S)$ takes 
$O(|S|+|\{B\in P~|~B\subseteq S\}|(|P|+ |\sra|))$  time.

\end{itemize}

\noindent
Let us prove that the overall number of newly generated blocks by
$\mathit{SplittingProcedure}()$ at line~6 of $\SSA$ is
$2(|P_{\mathrm{stsim}}|-|P_{\ell}|)$.  Let $\{P_i\}_{i\in [0,n]}$ be
the sequence of partitions computed by $\SSA$ where $P_0$ is the
initial partition $P_{\ell}$, $P_n$ is the final partition $\Pstsim$
and for all $i\in [0,n-1]$, $P_{i+1} \preceq P_{i}$. The number of
newly generated blocks by one splitting operation that refines $P_i$
to $P_{i+1}$ is given by $2(|P_{i+1}| - |P_i|)$. Thus, the overall
number of newly generated blocks is $\sum_{i=0}^{n-1} 2(|P_{i+1}| -
|P_i|) = 2(|P_{\mathrm{stsim}}|-|P_{\ell}|)$.

\noindent
It turns out that the overall number of iterations of the main
while-loop of $\SSA$ is in $O(|\Pstsim|^2)$.  If at some iteration of
$\SSA$ it happens that line~7 of $\mathit{Refine}()$ sets $\Rel(B,C)
:= \cff$ for some blocks $B$ and $C$ then for all the successive
iterations of $\SSA$, for any block $D$ which is contained in $B$
(namely, which is a descendant of $B$) and for any block $E$ which is
contained in $C$, and for all the successive iteratuons we will have
that $\Rel(D,E) = \cff$.  Moreover, at any iteration of $\SSA$, there
exist at least two blocks $B,C\in P$ such that the assignment
$\Rel(B,C):= \cff$ at line~7 of $\mathit{Refine}()$ is executed. Since
for any block $B$, the assignment $\Rel(B',C) :=\cff$ for some
$B'\subseteq B$ and for some $C$ may happen at most $|\Pstsim|$ times,
we obtain that the overall number of iterations is in
$O(|\Pstsim|^2)$.

\noindent
Hence,  the overall 
time complexities of the functions called within the main while-loop of
$\SSA$ are as follows:\\[5pt]
--~$\mathit{FindRefiner}()$: $O(|\Pstsim|^2(|\Pstsim|^2 +
  |\sra|+|\Pstsim||\sra^{\exists}|))$;
\\
--~$\mathit{Image}(\tuple{P,\Rel},B)$: $O(|\Pstsim|^2 |\Sigma|)$;
\\
--~$\pos(S,T)$: $O(|\Pstsim|^2(|\Sigma| + |\sra|))$;
\\ 
--~$\mathit{SplittingProcedure}(\tuple{P,\Rel},S)$:
  $O(|\Pstsim|^3 |\Sigma|)$. 
\\[5pt]
The analysis of the overall time complexity of
  $\mathit{Refine}(\tuple{P,\Rel},S)$ needs the following observation.
  As observed above, if at some iteration of $\SSA$ it happens that
line~7 of $\mathit{Refine}()$ sets $\Rel(B,C) := \cff$ for some blocks $B$ and $C$ 
then for all the successive iterations of
  $\SSA$, for any block $D$ which is contained in $B$  and
  for any block $E$ which is contained in $C$, we will have that $\Rel(D,E)
  = \cff$. Thus, for a given block $B$, if the test
  $\Rel(C,B)=\ctt$ at line~6 of $\mathit{Refine}()$ is true then for any block
  $C'$ which is descendant of $C$, the test $\Rel(C',B)=\ctt$ 
will be false. This
  means that for any given block $B$, the body at lines~7-16 of the
  if-then statement at line~6 will be executed at most $|\Pstsim|$
  times. Therefore, the overall time complexity in $\SSA$ of lines~3
  and~17 of
  $\mathit{Refine}(\tuple{P,\Rel},S)$ is $O(|\Pstsim|(1 + |\sra| +
  |\Pstsim| + |\sra|))=O(|\Pstsim|^2 + |\Pstsim||\sra|)$. Since the
  overall cost of lines~2-7 and 18 is $O(|\Pstsim|^2 (|\Sigma| +
  |\Pstsim|^2))$, it turns out that the overall cost of
  $\mathit{Refine}(\tuple{P,\Rel},S)$ is $O(|\Pstsim| (|\sra| +
  |\Pstsim||\Sigma| + |\Pstsim|^3))$.

\noindent
Summing up, the overall time complexity of $\SSA$ is 
$$O(|\Pstsim|^2(|\Sigma|+|\sra|+|\Pstsim|^2 +
|\Pstsim||\sra^\exists|)).$$
If $\sra$ is total then $|\Sigma|\leq |\sra|$ and $|\Pstsim|\leq
|\sra^\exists|$ so that the time complexity of $\SSA$ simplifies to 
$O(|\Pstsim|^2(|\sra|+
|\Pstsim||\sra^\exists|))$.

\medskip
\noindent
\textbf{Space Complexity.} 
The space complexity of $\SSA$ 
is in $O(|\Sigma||\Pstsim|\log|\Sigma|)$ because:

\begin{itemize}
\item[--] The pointers from any state $s\in
\Sigma$ to the block $P(s)$ of the current partition are
stored in 
$O(|\Sigma|\log|\Pstsim|)$ space.
\item[--] The lists localBottoms and bottomBlocks globally take,
  respectively, $O(|\Pstsim||\Sigma|)$ and $O(|\Pstsim|^2)$ space. 
\item[--] The current partition $P$ is stored in $O(|\Pstsim|)$ space.
\item[--] The current relation $\Rel$ is stored in $O(|\Pstsim|^2)$ space.
\item[--] The resizable tables Count and BCount take,
  respectively, $O(|\Sigma||\Pstsim|\log |\Sigma|)$ and
  $O(|\Pstsim|^2\log |\Sigma|)$ space.
\item[--] The local table Refiner in function $\mathit{FindRefiner}()$
  takes $O(|\Pstsim|^2)$. 
\qed
\end{itemize}

\decmargin{0.7em}
\linesnotnumbered
\begin{algorithm}[Htp]
\small
\SetVline
\SetAlTitleFnt{textsc}
\Indm

\FuncSty{$\mathit{Initialize()}\;\{$}

\Indp
$\mathit{CollapseSSC}(\Sigma,\sra,P)$\;
$\mathit{TopologicalSort}(\Sigma,\sra,P)$\;
$\mathit{TopologicalSort}(P,\Rel)$\;
\tcp{{\rm initialize Count}}
\ForAll{$y\in \Sigma$}{
  \ForAll{$x\in \pre(y)$}{
    \ForAll{$C\in P$}{\lIf{{\rm $(\Rel(y.\text{block},C) = \ctt)$}}{$\text{Count}(x,C)$++\;}}
  }
}
\tcp{{\rm initialize BCount}}
\ForAll{$C \in P$}{
      \lForAll{$x\in \Sigma$}{BCount$(x.\text{block},C)$ +$=$ Count$(x,C)$\;}
}
\tcp{{\rm initialize localBottoms and bottomBlocks}}
\ForAll{$B\in P$}{
    \lIf{{\rm $(\exists x\in B.\; \text{Count}(x,B) = 0)$}}{$B.\text{localBottoms}.\text{append}(x)$\;}
  \ForAll{$C\in P$ \KwSty{such that} $C\neq B$}{
    \If{{\rm $(\Rel(C,B) = \ctt ~\&~\exists x\in C.\; \text{Count}(x,B)=0)$}}
  {$B.\text{bottomBlocks}.\text{append}(C)$\;}
}
}

\Indm
\FuncSty{\}}

\BlankLine
\BlankLine
\clist{\cstate} \FuncSty{$\mathit{Image}(\PR\;\tuple{P,\Rel},$
  \cblock~$B)\;\{$}

\Indp
\clist{\cstate} $R := \varnothing$\;
  \ForAll{{\rm $C\in P$ \KwSty{such that} $(\Rel(C,B) = \ctt)$}}{
      \lForAll{{\rm $x\in C$}}{
        $R.\text{append}(x)$\;
      }
    }
\creturn $R$\;

\Indm
\FuncSty{\}}

\caption{$\mathit{Initialize}()$ and $\mathit{Image}()$ Functions.}\label{fig-2}
\end{algorithm}

\linesnotnumbered
\begin{algorithm}[!tp]
\small
\SetVline
\SetAlTitleFnt{textsc}
\Indm

\FuncSty{$\mathit{Update()}\;\{$}

\Indp
resize Count, BCount\;
    \ForAll{$B \in \text{split}$}{
      \ForAll{$x\in \Sigma$}{
        Count$(x,B.\text{intersection}):=\text{Count}(x,B)$\;}
      $B$.localBottoms $:=B$.localBottoms $\cap\, B$\;            
      $B$.intersection.localBottoms $:=$

      \hspace*{20ex}$B$.localBottoms $\cap\, B$.intersection\;
      \ForAll{$C \in P$}{
        \cint $k := \text{BCount}(B,C)$\; 
        BCount$(B,C):=0$\;
        \ForAll{$x\in B$}{$\text{BCount}(B,C)$
          +$= \text{Count}(x,C)$\;}
        BCount$(B.\text{split},C):= k ~\text{--}~ \text{BCount}(B,C)$\;
      }
    }

    \ForAll{$B\in P$}{
      \ForAll{{\rm $C\in B.\text{bottomBlocks}$ \KwSty{such that} 
          \hspace*{5.5ex}$C.\text{intersection} \neq \cnull$}}{
        \If{{\rm $(\forall x\in C.\; \text{Count}(x,B)>0)$}}
        {
          $B.\text{bottomBlocks}.\text{remove}(C)$\;
          $B.\text{bottomBlocks}.\text{append}(C.\text{intersection})$\;
        }
        \Else{
          \If{{\rm $(\exists x\in C.\text{intersection}.\; \text{Count}(x,B)=0)$}}
          {$B.\text{bottomBlocks}.\text{append}(C.\text{intersection})$\;}
        }
      }
    }
    \Indm
\FuncSty{\}}

\caption{$\mathit{Update}()$ Function.}\label{fig-bis}
\end{algorithm}

\end{document}